\documentclass[11pt]{article}
\usepackage[a4paper,scale={0.7,0.75}]{geometry}

\usepackage{amsmath,amscd} 
\usepackage{amsthm}
\usepackage{esvect}
\usepackage{graphicx}
\usepackage{amsfonts}
\usepackage{float}
\usepackage{mdwlist}
\usepackage[all]{xy}
\usepackage{color}
\usepackage[marginpar]{todo}
\usepackage{enumerate}
\usepackage{xspace}
\usepackage[capitalize]{cleveref}
\usepackage{cite}

\newtheorem{thm}{Theorem}[section]

\newtheorem{prop}[thm]{Proposition}
\newtheorem{lem}[thm]{Lemma} 
\newtheorem{cor}[thm]{Corollary}
\newtheorem{conjecture}[thm]{Conjecture}

\theoremstyle{definition}
\newtheorem{example}[thm]{Example}

\theoremstyle{definition}
\newtheorem{remark}[thm]{Remark}

\long\def\symbolfootnote[#1]#2{\begingroup%
\def\thefootnote{\fnsymbol{footnote}}\footnote[#1]{#2}\endgroup}

\DeclareMathOperator{\Rips}{\mathrm{Rips}}

\DeclareMathOperator{\obj}{obj\,}
\newcommand{\gr}[0]{\mathrm{gr}}

\newcommand{\free}[0]{\mathrm{fr}}
\newcommand{\fl}[0]{\mathrm{fl}}

\newcommand{\z}[0]{a}
\newcommand{\y}[0]{a}
\renewcommand{\b}[0]{b}
\newcommand{\Gen}[0]{g}

\DeclareMathOperator{\im}{im}

\newcommand{\Z}[0]{\mathbb Z}
\newcommand{\R}[0]{\mathbb R}
\newcommand{\RCat}[0]{\mathbf R}
\newcommand{\Top}[0]{\mathbf{Top}}
\newcommand{\Vect}[0]{\mathbf{Vect}}
\newcommand{\NN}[0]{\mathbb N}
\newcommand{\Q}[0]{\mathbb Q}
\newcommand{\A}[0]{\mathcal A}
\newcommand{\B}[0]{\mathcal B}
\newcommand{\C}[0]{\mathcal C}
\newcommand{\D}[0]{\mathcal D}
\newcommand{\E}[0]{\mathcal E}
\newcommand{\F}[0]{\mathcal F}

\newcommand{\V}[0]{\mathcal V}
\newcommand{\I}[0]{\mathcal I}

\newcommand{\ex}[0]{\mathrm Ex}
\newcommand{\longsup}[0]{u}

\newcommand{\translation}[2]{\varphi_{#1}^{#2}}
\renewcommand{\cong}[0]{\simeq}

\newcommand{\nmod}{n\mathbf{\textup{-}Mod}}
\newcommand{\nfilt}{n\mathbf{\textup{-}Filt}}
\newcommand{\CS}{n\mathbf{\textup{-}Fun}}
\newcommand{\onemod}{1\mathbf{\textup{-}Mod}}
\newcommand{\PFmod}{\mathrm{f1}\mathbf{\textup{-}Mod}}

\newcommand{\id}{\mathrm{Id}}

\newcommand{\cells}[0]{\mathrm{cells}}

\newcommand{\W}[0]{\mathcal W}
\newcommand{\Y}[0]{\mathcal Y}

\newcommand{\FS}[0]{{\mathcal S}}

\newcommand{\newspan}[0]{{\rm span}\,}
\renewcommand{\th}[0]{\mathrm{th}}

\renewcommand{\subsubsection}[1]{\paragraph{#1\rm{.}}}
\newcommand{\nparagraph}[1]{\paragraph{#1\rm{.}}}

\newcommand{\pfd}{p.f.d.\@\xspace}
\begin{document}


  \title{The Theory of the Interleaving Distance on Multidimensional Persistence Modules}\date{}
\author{Michael Lesnick\footnote{mlesnick@ima.umn.edu.
} \\ Institute for Mathematics and its Applications
}
\maketitle
\symbolfootnote[0]{AMS Subject Classification 55, 68.}
\symbolfootnote[0]{Key phrases: multidimensional persistence, stability of persistent homology, persistence modules,}
\symbolfootnote[0]{interleavings, algebraic stability, isometry theorem.}
{\abstract {In 2009, Chazal et al. introduced \emph{$\epsilon$-interleavings} of persistence modules.  $\epsilon$-interleavings induce a pseudometric $d_I$ on (isomorphism classes of) persistence modules, the {\it interleaving distance}.  The definitions of $\epsilon$-interleavings and  $d_I$ generalize readily to multidimensional persistence modules.  In this paper, we develop the theory of multidimensional interleavings, with a view towards applications to topological data analysis.  

We present four main results.  First, we show that on 1-D persistence modules, $d_I$ is equal to the bottleneck distance $d_B$.  This result, which first appeared in an earlier preprint of this paper, has since appeared in several other places, and is now known as the \emph{isometry theorem}.

\begin{sloppypar}
Second, we present a characterization of the $\epsilon$-interleaving relation on multidimensional persistence modules.  This expresses transparently the sense in which two $\epsilon$-interleaved modules are algebraically similar.  
\end{sloppypar}
Third, using this characterization, we show that when we define our persistence modules over a prime field, $d_I$ satisfies a universality property.  This universality result is the central result of the paper.  It says that $d_I$ satisfies a stability property generalizing one which $d_B$ is known to satisfy, and that in addition, if $d$ is any other pseudometric on multidimensional persistence modules satisfying the same stability property, then $d\leq d_I$.  We also show that a variant of this universality result holds for $d_B$, over arbitrary fields. 

Finally, we show that $d_I$ restricts to a metric on isomorphism classes of finitely presented multidimensional persistence modules. 

}}   

\section{Introduction}

\subsection{Background and Motivation}\label{Sec:Background}
\paragraph{Persistent Homology}

Persistent homology is a topological tool for studying the global, non-linear, geometric features of data.  In the last decade and a half, it has been applied widely \cite{edelsbrunner2010computational,Carlsson2014Topological} and has been the subject of a large body of theoretical work.  

Persistent homology provides algebraic invariants, called \emph{persistence modules}, of a variety of types of data, including finite metric spaces and $\R$-valued functions.  Let $\RCat$ denote the poset category of real numbers, and let $\Vect$ denote the category of vector spaces over some fixed field $k$.  We define a persistence module to be a functor $\RCat\to \Vect$.

To construct a persistence module from a data set, we first associate to our data a \emph{filtration}, i.e., a functor $\F:\RCat\to \Top$ that maps each element of $\hom(\RCat)$ to an inclusion.  For example, if our data is an $\R$-valued function $\gamma:T\to \R$, for $T$ a topological space, we may take $\F$ to be the \emph{sublevelset filtration $\FS(\gamma)$}, defined by
\[\FS(\gamma)_t=\{y\in T\mid \gamma(y)\leq t\},\quad t\in \R.\]
Since $\FS(\gamma)_s\subset \FS(\gamma)_t$ whenever $s\leq t$, this indeed gives a filtration.

Letting $H_i:\Top\to \Vect$ denote the $i^{th}$ singular homology functor with coefficients in $k$, we obtain a persistence module $H_i \F$ for any $i\geq 0$.  $H_i \F$ algebraically encodes geometric information about our data.

 The structure theorem for persistence modules of \cite{crawley2012decomposition} tells us that if $M$ is a persistence module whose vector spaces are finite dimensional, then $M$ decomposes in an essentially unique way into simple indecomposables called \emph{interval persistence modules}; see \cref{Sec:IsometrySection}.  The interval persistence modules are parameterized by the non-empty intervals in $\R$.  Hence, we may associate to $M$ a collection of intervals $\B_M$ which indexes the indecomposables of $M$.  We call $\B_M$ the \emph{barcode of $M$}.  

We usually work with the persistence modules $H_i \F $ by way of their barcodes; we regard each interval in $\B_{H_i \F}$ as a topological feature of our data, and we interpret the length of the interval as a measure of significance of that feature.  

\paragraph{The Bottleneck Distance}
Both theory and applications of persistent homology make extensive use of pseudometrics on barcodes.  The \emph{bottleneck distance} $d_B$, a readily computed \cite{cohen2007stability} and particularly well-behaved pseudometric on barcodes, is the most common choice.  We give the definition of $d_B$ in \cref{1DPreliminaries}. 
Stability theorems for persistent homology \cite{cohen2007stability,chazal2009proximity,chazal2009gromov} and theorems about inferring persistent homology from point cloud data \cite{chazal2009analysis, chazal2009persistence} are usually formulated using $d_B$.  Moreover, many applications of persistent homology to shape comparison and related tasks \cite{verri1993use,carlsson2004persistence,carlsson2006algebraic,biasotti2008multidimensional,cohen2010lipschitz} rely in an essential way on computations of $d_B$ and its variants.

Via the correspondence between persistence modules and their barcodes, we can regard $d_B$ as a pseudometric on persistence modules.  

\subsubsection{Multidimensional Persistent Homology}

In 2006, the authors of \cite{carlsson2009theory} introduced a multidimensional generalization of persistent homology.  Whereas ordinary persistent homology maps filtrations to persistence modules, multidimensional persistent homology maps \emph{$n$-dimensional filtrations} to \emph{$n$-dimensional persistence modules}, in essentially the same way.

For $n\geq 1$, define a partial order on $\R^n$ by taking $(a_1,\ldots,a_n)\leq (b_1,\ldots,b_n)$ if and only if $a_i\leq b_i$ for all $i$, and let $\RCat^n$ denote  the associated poset category.  We define an $n$-dimensional filtration $\F$ to be a functor $\RCat^n\to \Top$ that maps each element of $\hom(\RCat^n)$ to an inclusion, and we define an $n$-dimensional persistence module to be a functor $\RCat^n\to \Vect$.  In \cref{Sec:PersistenceModules}, we give an equivalent definition of an $n$-D persistence module as a graded module over a  monoid ring.  

As in the 1-D case, for $\F$ an $n$-dimensional filtration, $H_i\F$ is an $n$-dimensional persistence module.

$n$-dimensional filtrations, $n>1$, arise naturally from data in a number of ways.  We briefly describe two of these ways.

\begin{example}\label{Ex:SublevelsetFiltration}
As observed in \cite{frosini1999size,cerri2009multidimensional}, any function $\gamma:T\to \R^n$ on a topological space $T$ gives rise to an \emph{$n$-dimensional sublevelset filtration} $\FS(\gamma)$, defined by \[\FS(\gamma)_a=\{y\in T\mid \gamma(y)\leq a\},\quad a\in \R^n.\]
This generalizes the sublevelset filtration introduced above in the case $n=1$. 

Frequently in topological data analysis, we have several $\R$-valued functions \[\gamma_1,\gamma_2,\ldots,\gamma_n:T\to \R,\] which we can regard as a single function $\gamma:T\to \R^n$.  For example, such ensembles of functions arise naturally in shape and image classification  applications \cite{biasotti2008multidimensional,li2013persistence,adcock2014classification}, \cite[Section 5.2]{Carlsson2014Topological} and in the study of time-varying data \cite{cohen2006vines}.  In these settings, the multidimensional persistent homology of $\FS(\gamma)$ generally encodes much more information about the ensemble of functions $\gamma$ than can be encoded using 1-D persistent homology.
\end{example}

\begin{example}
As explained in \cite{carlsson2009theory}, multidimensional filtrations also arise naturally as invariants of metric spaces.  In topological data analysis, we often study a finite metric space $P$ by considering its \emph{Vietoris-Rips filtration} $\Rips(P)$.  This is the 1-dimensional simplicial filtration defined by taking $\Rips(P)_t$ to be the maximal simplical complex with 0-skeleton $P$ and 1-simplices the edges $[p,q]$ with $d(p,q)\leq 2t$.  

While for any $i\geq 0$, $H_i\Rips(P)$ is stable to perturbations of the metric \cite{chazal2009gromov}, it is highly unstable to the addition and removal of outliers in $P$.  Relatedly, $H_i\Rips(P)$ is insensitive to variation in the density of points in $P$.  

To address these issues, \cite{carlsson2008local} and \cite{carlsson2009theory} suggested that we consider a \emph{codensity function} \[\gamma:P\to \R,\] a function on $P$ whose value is low at dense points and high at outliers \cite{carlsson2009theory,wasserman2004all}.  $\gamma$ can be defined using the metric structure on $P$ alone, using for example a $k$-nearest neighbors density estimate \cite{carlsson2009theory}.  Given $\gamma$, we may define a 2-dimensional filtration $\Rips(\gamma)$ by taking
\[\Rips(\gamma)_{(a,b)}=\Rips(\gamma^{-1}(-\infty,a])_b.\] 
The 2-D persistence modules $H_i\Rips(\gamma)$ are more robust to noise and more sensitive to variations of density than the 1-D persistence modules $H_i\Rips(P)$.  (See also \cite{chazal2011geometric}, which constructs a different 2-D filtration from point cloud data in way that is also robust to noise.)  

More generally, any function $\gamma:P\to \R^{n}$ yields an $(n+1)$-dimensional filtration $\Rips(\gamma)$, in essentially the same way.  There are several other interesting functions $\gamma$ which we can define from the metric structure on $P$ alone.  For example, as suggested in \cite{carlsson2009theory}, we may take $\gamma:P\to \R$ to be an \emph{eccentricity function}, some measure of centrality of points in $P$.  Or we may take $\gamma:P\to \R$ to simply be the distance to a fixed point $p\in P$.  When $\gamma$ is induced from the metric on $P$, $\Rips(\gamma)$ topologically encodes information about the metric; for different choices of $\gamma$, $\Rips(\gamma)$ encodes different kinds of information.  
\end{example}

\paragraph{The Difficulty of Defining Barcodes for Multi-D Persistence Modules}
While multidimensional persistence modules are far richer invariants than their 1-dimensional counterparts, they also are far more complex.  As a consequence, the definition of a barcode does not extend to multidimensional persistence modules in any completely satisfactory way.  
 
A finitely presented $n$-D persistence module $M$ can be written in an essentially unique way as a direct sum of indecomposables; this follows easily from a standard formulation of the Krull-Schmidt theorem \cite{atiyah1956krull}.  Thus, in principle, we can define the barcode $\B_M$ of $M$ as we do in the 1-D case, as the collection of isomorphism classes of the indecomposables of $M$.  However, it follows easily from Gabriel's theorem in the theory of quiver representations \cite{gabriel1972unzerlegbare, derksen2005quiver} that for $n>1$, the set of isomorphism classes of finitely presented, indecomposable, $n$-D persistence modules is extremely complicated.  Thus, the invariant $\B_M$ will generally be far too complicated to be useful in the ways that the barcode of a 1-D persistence module is typically useful.  In particular, there seems to be no naive way to define a multidimensional generalization of the bottleneck distance in terms of such generalized barcodes $\B_M$.

Naively, one might hope to give a different definition of the barcode of $M$, as a collection of nice subsets of $\R^n$, much as the barcode of a 1-D persistence module is defined as a collection of nice subsets of $\R$.  However, in a sense that can be made precise, there is no good way of formulating such a definition when $n>1$, even if we allow our invariant to be incomplete. 
 
\subsubsection{Distances on Multidimensional Persistence Modules}
The main goal of this paper is to show that, in spite of the unavailability of a fully satisfactory definition of the barcode for multi-D persistence modules, the bottleneck distance $d_B$ does admit a simple and very well-behaved generalization to the multidimensional setting.  This generalization, the \emph{interleaving distance} $d_I$, is defined directly on persistence modules, in a ``barcode-free'' way.

The question of how to best generalize $d_B$ to the setting of multi-D setting is one of basic importance to the theory of multidimensional persistent homology: In order to adapt to the multidimensional setting the many theoretical results for 1-D persistence which are formulated using $d_B$, we require a good multidimensional generalization of $d_B$ \cite{lesnick2012multidimensional}.

A number of papers have introduced pseudometrics on multidimensional persistence modules \cite{ishkhanov2008topological,carlsson2010multiparameter,cerri2009multidimensional,frosini2010stable}, and several of these have presented stability results for the metrics they introduce \cite{carlsson2010multiparameter,cerri2009multidimensional,frosini2010stable}.  The multi-dimensional matching distance of \cite{cerri2009multidimensional} is unique amongst the pseudometrics introduced by these papers in that it is a generalization of $d_B$.

However, in choosing a multidimensional generalization of $d_B$ for use in the development of theory or in applications, we want more of our distance than just good stability properties.  A stability result of the kind typically appearing in the persistent homology literature \cite{cohen2007stability,chazal2009proximity,cohen2010lipschitz,cerri2009multidimensional} tells us that our distance on persistence modules is not, in some relative sense, too sensitive.  On the other hand, a good choice of distance should also not be too insensitive.  As an extreme illustration of this, consider the pseudometric on persistence modules which is identically 0; it satisfies lots of strong stability properties, yet is clearly too insensitive to be of any use.

Ideally, then, we would like to have a generalization of $d_B$ to multidimensional persistence modules which is not only stable, but also is as sensitive as a stable metric can be, in a suitable sense.  Our universality result shows that $d_I$ satisfies a property of this kind, and that it is the unique distance which does so.  The result thus distinguishes $d_I$ from the many possible choices of stable distances on multidimensional persistence modules as a particularly natural choice for use in the development of theory.  

\subsection{Overview of Results}
We now give an overview of our main results, expanding on the abstract.  Where convenient, we refer to multidimensional persistence modules simply as ``persistence modules."

\subsubsection{The Isometry Theorem}

Our first main result, Theorem~\ref{InterleavingEqualsBottleneck}, shows that on 1-D persistence modules whose vector spaces are finite dimensional, $d_I=d_B$.  In view of the algebraic stability theorem of \cite{chazal2009proximity}, which says that $d_I\geq d_B$, it is enough for us to show that $d_I\leq d_B$.  Adopting the terminology of \cite{bubenik2012categorification}, we call the result that $d_I=d_B$ the {\it isometry theorem}.  

Our proof of the isometry theorem relies on a recent version of the structure theorem for 1-D persistence modules, due to Crawley-Boevey \cite{crawley2012decomposition}.  An earlier version of the present paper \cite{lesnick2011optimalityV2}, written before \cite{crawley2012decomposition} was available,  proved the isometry theorem using a weaker version of the structure theorem, due to Webb \cite{webb1985decomposition}.  The stronger structure theorem of \cite{crawley2012decomposition} allows for major simplifications both in the definition of $d_B$ for 1-D persistence modules indexed by $\R$ and in our proof that $d_I=d_B$.  In fact, given the algebraic stability theorem of \cite{chazal2009proximity} and the structure theorem of \cite{crawley2012decomposition}, the proof of the isometry theorem becomes almost trivial.

\nparagraph{Other Proofs of the Isometry Theorem} Several months after I posted the first version of this paper to the arXiv in 2011, two papers \cite{bubenik2012categorification,chazal2012structure} were posted to the arXiv, each of which also proves a version of the isometry theorem.   The version of the result presented in \cite{bubenik2012categorification} is a special case of our result.  The version of the result presented in \cite{chazal2012structure} is a variant of our result which applies to a more general class of persistence modules called q-tame persistence modules.  See also \cite{bauer2014induced} for a recent proof of the isometry theorem in the q-tame setting which avoids use of some of the technical machinery of \cite{chazal2012structure}.

As explained in \cite{chazal2012structure}, the structure theorem of \cite{crawley2012decomposition} does not fully extend to q-tame persistence modules.  Given this, the definition of the barcode in the q-tame setting requires some care; see \cite{chazal2014observable} for a recent investigation of this matter.  
 
\subsubsection{The Characterization of $\epsilon$-Interleaving Relation}

Our second main result is Theorem~\ref{AlgebraicRealization}, a characterization of $\epsilon$-interleaving relation on multidimensional persistence modules; it expresses transparently the sense in which two $\epsilon$-interleaved persistence modules are algebraically similar.  The result tells us that two persistence modules are $\epsilon$-interleaved if and only if there exist presentations for the two modules that are similar, in the sense that they differ from one another by small shifts in the grades of the generators and relations.  The result in turn yields a characterization of $d_I$.  

Our characterization of $\epsilon$-interleaved pairs of modules in fact holds, with essentially the same proof, for more general types of interleavings between multidimensional persistence modules; see \cite{lesnick2012multidimensional} for the more general result and an application of it to topological inference using Vietoris-Rips complexes.

\subsubsection{The Universality of $d_I$}
Our third main result, Corollary~\ref{CorOptimality}, is our universality result for $d_I$.  It tells us that for multidimensional persistence modules over a prime field (i.e., $\Q$ or $\Z/p\Z$ for $p$ prime), 
\begin{enumerate}[(i)]
\item $d_I$ is stable in a sense analogous to that in which the bottleneck distance is shown to be stable in \cite{cohen2007stability,chazal2009proximity},
\item if $d$ is another pseudometric on multidimensional persistence modules, then $d\leq d_I$.
\end{enumerate}
We can interpret this result as the statement that $d_I$ is the terminal object in a poset category of stable metrics on multidimensional persistence modules.

This universality result is novel even for 1-D persistence.  In that case, it offers some mathematical justification, complementary to that of \cite{cohen2007stability,chazal2009proximity}, for the use of the bottleneck distance.  

In fact, provided we restrict attention to 1-D persistence modules whose vector spaces are finite dimensional, we may dispense with the assumption that our field of coefficients is prime; our Theorem~\ref{PFOptimality} gives an analogue of Corollary~\ref{CorOptimality} for this class of modules, over arbitrary fields.

The main ingredient in the proof of Corollary~\ref{CorOptimality} is our characterization result, Theorem~\ref{AlgebraicRealization}.  Using that result, we present a constructive argument which shows that when the underlying field is prime, $\epsilon$-interleavings can, in a suitable sense, be lifted to a category of $\R^n$-valued functions. Given this, our universality result follows readily.
  
After I posted the first version of this paper to the arXiv, I learned that for the special case of $0^{\th}$ 1-D persistent homology, d'Amico et al.  \cite{d2010natural} had previously established a universality result for $d_B$ similar to the universality results given here.  The main universality result of that work, \cite[Theorem 32]{d2010natural}, is a special case of our Theorem~\ref{PFOptimality}.  
  
\subsubsection{The Closure Theorem}

Our fourth main result, Theorem~\ref{InterleavingThm}, says that for two finitely presented multidimensional persistence modules $M$ and $N$, if $d_I(M,N)=\epsilon$, then $M$ and $N$ are $\epsilon$-interleaved.  This is equivalent to the statement that the set 
\[\{\epsilon\mid M\textup{ and  }N\textup{ are }\epsilon\textup{-interleaved}\}\] 
is closed in $\R$.  We thus call this result the {\it closure theorem.}  Considering the case $\epsilon=0$, it follows that $d_I$ restricts to a metric on isomorphism classes of finitely presented multidimensional persistence modules.  

\subsection{Computation of the Interleaving Distance}
While it is known that the bottleneck distance on 1-D persistence modules can be computed readily \cite{cohen2007stability}, the question of if and how the interleaving distance on $n$-D persistence modules can be efficiently computed remains open for $n>1$.  

An earlier preprint of this paper and my Ph.D. thesis \cite{lesnick2012multidimensional} presented partial results on computation of the interleaving distance, which for brevity's sake are omitted here: It was shown that given $\epsilon\geq 0$ and presentations of $n$-D persistence modules $M$ and $N$ with a total of $m$ generators and relations, deciding whether $M$ and $N$ are $\epsilon$-interleaved is equivalent to deciding whether a solution exists to a certain system of multivariate quadratics with $O(m^2)$ equations and $O(m^2)$ variables.  Further, it was shown that $d_I(M,N)$ lies in a set $S\subset [0,\infty)$ of size $O(m^2)$, determined in a simple way from the presentations of $M$ and $N$.  Hence, performing a binary search over $S$, we can in principle compute $d_I(M,N)$ by deciding whether $M$ and $N$ are $\epsilon$-interleaved for $O(\log m)$ values of $\epsilon$.

However, the general problem of deciding whether a solution exists to a system of quadratics is NP-complete.  We do not yet understand the complexity of deciding whether solutions exist to the specific systems of quadratics that arise in our setting.

\subsection{Other Work on Generalized Interleavings}
Interleavings and interleaving distances can also be defined on multidimensional filtrations.  My Ph.D. thesis \cite{lesnick2012multidimensional} presents results on interleavings between multidimensional filtrations, some of which parallel the results  presented here.  Building on that work, Andrew Blumberg and I have shown that an analogue of the main universality result of the present paper holds for a homotopy theoretic variant of the interleaving distance on multidimensional filtrations \cite{blumberg2014universality}.  

In addition, \cite{lesnick2012multidimensional} studies applications of multidimensional interleavings and interleaving distances to topological inference.  In particular \cite[Chapter 4]{lesnick2012multidimensional}, presents multidimensional analogues of a topological inference theorem of Chazal, Guibas, Oudot, and Skraba \cite{chazal2009persistence}, formulated directly on the level of filtrations.  The inference theorem of \cite{chazal2009persistence}, which we can think of as a loose analogue of the weak law of large numbers for persistent homology \cite[Section 1.2.2]{lesnick2012multidimensional}, adapts readily to the multidimensional setting, given the language of interleavings.

Since the first version of this paper was posted to the arXiv, several other authors have studied interleavings and interleaving distances in various generalized persistence settings: \cite{curry2013sheaves} considers interleavings on (co)presheaves; \cite{morozov2013interleaving} and \cite{de2014categorification} study interleaving distances on merge trees and Reeb graphs; and \cite{bubenik2013metrics} introduces and studies a general definition of the interleaving distance on diagrams indexed by an arbitrary preordered metric space.  

It would be interesting to know whether the universality result of this paper adapts to these other settings.  

\subsection{Organization of the Paper}
The paper is organized as follows.  Section~\ref{SectionAlgebraicPreliminaries} covers preliminaries that will be needed for the rest of the paper.  In particular, we give the module-theoretic definition of a multidimensional persistence module and define the interleaving distance.

Each of Sections~\ref{Sec:IsometrySection}-\ref{Sec:ClosureTheorem} centers on one of our main results.  These sections 
can largely be read independently of each other, with two exceptions: The proof of our main universality result in Section~\ref{Sec:OptimalitySection} depends on our characterization of the $\epsilon$-interleaving relation in Section~\ref{Sec:Characterization}, and our treatment of the  closure theorem in \ref{Sec:ClosureTheorem} uses presentations of multi-D persistence modules, discussed in \cref{Sec:Presentations}.

The paper concludes in Section~\ref{DiscussionSection} with a discussion of open problems and future directions for research.
  \section{Preliminaries}\label{SectionAlgebraicPreliminaries}
\subsection{Multidimensional Persistence Modules}\label{Sec:PersistenceModules}
In \cref{Sec:Background}, we defined an $n$-dimensional persistence module as a diagram of vector spaces indexed by $\RCat^n$.  In fact, as we now explain, we can interpret this object as a module in the usual algebraic sense.  This interpretation will be convenient in Sections~\ref{Sec:Characterization}-\ref{Sec:ClosureTheorem}, when we work with presentations of multidimensional persistence modules. 

\subsubsection{The Ring $P_n$}\label{MonoidRings}
For $k$ a field, let the ring $P_n$ be the analogue of the usual polynomial ring $k[x_1,\ldots,x_n]$ in $n$ variables, where exponents of the indeterminates in $P_n$ are allowed to take on arbitrary values in $[0,\infty)$ rather than only values in the non-negative integers.  For example, if $k=\Q$ then 
\[1+x_2+x_1^\pi+\frac{2}{5}x_1^3x_2^{\sqrt 2}\in P_2.\]
Formally, $P_n$ can be defined as a {\it monoid ring} over the monoid $([0,\infty)^n,+)$ \cite{lang2002algebra}.

For $a=(a_1,\ldots,a_n)\in [0,\infty)^n$, we let $x^a$ denote the monomial $x_1^{a_1}x_2^{a_2}\cdots x_n^{a_n}\in P_n$. 

\subsubsection{Module-Theoretic Description of a Multidimensional Persistence Module}
We define an \emph{$n$-graded module} (or \emph{n-module}, for short) to be a $P_n$-module $M$ with a direct sum decomposition as a $k$-vector space $M \cong \bigoplus_{a \in {\R}^n} M_a$, such that \[x^b(M_a) \subset M_{a+b}\] for all $a\in \R^n,$ $b\in [0,\infty)^n$.  

For $a\leq b\in \R^n$, the action of $x^{b-a}$ on $M$ defines a linear map $M_a\to M_b$, which we denote by $\varphi_M(a,b)$ and call a \emph{transition map}.  Note that for $a\leq b\leq c\in\R^n$, the following diagram commutes:
\[
\xymatrixcolsep{3.5pc}
\xymatrix{
& M_c\\
M_a\ar[r]_{\varphi_M(a,b)}\ar[ur]^{\varphi_M(a,c)} &M_b\ar[u]_{\varphi_M(b,c)}
}
\]

We define a category $\nmod$ of $n$-modules by taking the morphisms in $\nmod$ to be the module homomorphisms $f:M\to N$ such that $f(M_a)\subset N_a$ for all $a \in \R^n$.  We let $f_a:M_a\to N_a$ denote the restriction of the morphism $f$.  Note that for all $a\leq b\in \R^n$, the following diagram commutes:
\[
\xymatrixcolsep{3.5pc}
\xymatrix{
M_a\ar[r]^{\varphi_M(a,b)}\ar[d]_{f_a} &M_b\ar[d]^{f_b}\\
N_a\ar[r]_{\varphi_N(a,b)} &N_b
}
\]

\begin{remark}\label{Rmk:CategoryDef}
Let $\Vect^{\RCat^n}$ denote the category whose objects are $n$-D persistence modules (functors $\RCat^n\to \Vect$), as defined in \cref{Sec:Background}, and whose morphisms are natural transformations.  
There is an obvious isomorphism between $\nmod$ and $\Vect^{\RCat^n}$.  Thus we may identify the two categories, and work interchangeably with $n$-dimensional persistence modules and $n$-modules. 
\end{remark}

\subsubsection{Interval $n$-Modules}\label{Sec:IntervalModules}
We say $I\subset \R^n$ is a {\it (generalized) interval} if $I$ is non-empty and $a,c\in I$ implies $b\in I$ whenever $a\leq b\leq c$.  
For $I \subset \R^n$ an interval, define the $n$-module $C(I)$ by
\begin{equation*}
C(I)_a=
\begin{cases}
k &{\textup{if }} a\in I, \\
0 &{\textup{ otherwise}.}
\end{cases}
\end{equation*}
\vskip14pt
\begin{equation*}
\varphi_{C(I)}(a,b)=
\begin{cases}
\id_k &{\textup{if }} a,b\in I,\\
0 &{\textup{ otherwise}.}
\end{cases}
\end{equation*}
We refer to the module $C(I)$ as an {\it interval $n$-module.}

\subsubsection{Homogeneity}
Let $M$ be an $n$-module.  For $a\in \R^n$, we refer to any non-zero $v\in M_a$ as a {\it homogeneous element} of grade $a$. 
 A {\it homogeneous submodule} of an $n$-module is a submodule generated by a set of homogeneous elements.  The quotient of an $n$-module $M$ by a homogeneous submodule of $M$ is itself an $n$-module; the $n$-graded structure on the quotient is induced by that of $M$.

\subsection{$\epsilon$-Interleavings and the Interleaving Distance}\label{InterleavingsOfModules}
For $\epsilon\in \R$, let ${\vec \epsilon}\in \R^n$ denote the vector whose components are each $\epsilon$.  

\subsubsection{Shift Functors}\label{ShiftsOfModules} For $v\in \R^n$, we define the \emph{$v$-shift functor} \[(\cdot)(v):\nmod\to \nmod\] as follows: For $M$ an $n$-module, define the $n$-module $M(v)$ by taking $M(v)_a=M_{a+v}$ and $\varphi_{M(v)}(a,b)=\varphi_M(a+v,b+v)$ for $a\leq b\in \R^n$.
For $f$ a morphism in $\nmod$, define $f(v)$ by taking $f(v)_a=f_{a+v}$.

To keep notation simple, for $\epsilon\in \R$, we write $(\cdot)(\vec \epsilon):\nmod\to \nmod$ simply as $(\cdot)(\epsilon)$.

\subsubsection{Transition Morphisms}
For an $n$-module $M$ and $\epsilon\in [0,\infty)$, let \[\varphi_M^\epsilon:M\to M(\epsilon),\] the \emph{(diagonal) $\epsilon$-transition morphism}, be the morphism whose restriction to $M_a$ is the linear map $\varphi_M(a,a+\vec\epsilon)$ for all $a\in \R^n$.  

\subsubsection{$\epsilon$-Interleavings}
For $\epsilon\geq 0$, we say that two $n$-modules $M$ and $N$ are \emph{$\epsilon$-interleaved} if there exist morphisms $f:M\to N(\epsilon)$ and $g:N\to M(\epsilon)$ such that 
\begin{align*}
g(\epsilon) \circ f&=\varphi_M^{2\epsilon}\textup{ and }\\
f(\epsilon) \circ g&=\varphi_N^{2\epsilon}; 
\end{align*}
we call $f$ and $g$ \emph{$\epsilon$-interleaving} morphisms.

The definition of $\epsilon$-interleaving morphisms was introduced for $1$-modules in \cite{chazal2009proximity}.  See also \cite{bubenik2012categorification} for a rephrasing of the definition using the language of natural transformations; this is given for $1$-modules but extends immediately to $n$-modules.

\begin{remark}\label{EasyInterleavingRemark}
It's easy to show that if $0\leq\epsilon_1 \leq \epsilon_2$ and $M$ and $N$ are $\epsilon_1$-interleaved, then $M$ and $N$ are $\epsilon_2$-interleaved. 
\end{remark}

\subsubsection{Metrics and Pseudometrics}
Recall that an {\it extended pseudometric} on $X$ is a function $d:X\times X\to [0,\infty]$ with the following three properties:
\begin{enumerate*}
\item $d(x,x)=0$ for all $x\in X$.
\item $d(x,y)=d(y,x)$ for all $x,y\in X$.
\item $d(x,z)\leq d(x,y)+d(y,z)$ for all $x,y,z\in X$ with $d(x,y),d(y,z)< \infty$.
\end{enumerate*}
In this paper, by a {\it distance}, we will mean an extended pseudometric.

An {\it extended metric} is an extended pseudometric $d$ with the additional property that $d(x,y)\ne 0$ whenever $x\ne y$.  In what follows, we'll drop the modifier ``extended," and refer to extended (pseudo)metrics simply as (pseudo)metrics.

\subsubsection{The Interleaving Distance on $n$-modules}
For a category $C$, let $\obj C$ denote the objects of $C$ and let $[\obj C]$ denote the collection of isomorphism classes of objects of $C$.  For $M\in \obj C$, let $[M]$ denote the isomorphism class of $M$.  For $d$ a pseudometric on $[\obj C]$ and $M,N\in \obj C$, we write $d(M,N)$ as shorthand for $d([M],[N])$.

We define $d_I:[\obj\nmod]\times [\obj\nmod]\to [0,\infty]$, the \emph{interleaving distance}, by taking 
\[d_I(M,N)=\inf\, \{\epsilon\in [0,\infty)\mid M\textup{ and  }N\textup{ are }\epsilon\textup{-interleaved}\}.\]

Note that $d_I$ is a pseudometric.  However, the following example shows that $d_I$ is not a metric.  

\begin{example}\label{NonIsoSameDiagram}  Let $M$ be the $1$-module with $M_0=k$ and $M_{a}=0$ if $a\ne 0$.  Let $N$ be the trivial $1$-module.  Then $M$ and $N$ are not isomorphic, and so are not $0$-interleaved, but it is easy to check that $M$ and $N$ are $\epsilon$-interleaved for any $\epsilon>0$.  Thus $d_I(M,N)=0$. 
\end{example}

To offer some intuition for how $d_I$ behaves, we consider an additional example.

\begin{example}
For $I$ an interval, as defined in Section~\ref{Sec:IntervalModules}, define $w(I)$, the {\it width} of the interval $I$, by \[w(I)=\sup\, \{\epsilon\in [0,\infty)\mid \exists\, a\in I\textup{ such that }a+\vec\epsilon\in I\}.\]  If $N$ is the trivial $n$-module, it is easily checked that $d_I(C(I),N)=\frac{w(I)}{2}$.  For example, if $n=2$ and \[I_1=\{(s,t)\in \R^2\mid(0,0)\leq (s,t)<(2,2)\}\] then $d_I(C(I_1),N)=1$; similarly, if \[I_2=\{(x,y)\in\R^2\mid x\in [0,2)\textup{ or }y\in[0,2)\}\] then $d_I(C(I_2),N)=1$.
\end{example}

  \section{The Isometry Theorem}\label{Sec:IsometrySection}
In this section we present the isometry theorem, our first main result.

\subsection{Preliminaries for 1-D Persistence Modules}\label{1DPreliminaries}
\subsubsection{Basics} 
Informally, a \emph{multiset} is a set where an element can appear multiple times.  For our purposes it will be sufficient to restrict attention to multisets where each element can appear at most countably many times.  Formally, then, we may define a multiset $\A$ with underlying set $A$ to be a subset of $A\times \NN$ such that if $(s,n)\in \A$ and $n'\leq n\in \NN$ then $(s,n')\in \A$.

Let ${\mathcal I}$ denote the set of all (non-empty) intervals in $\R$.  We define a \emph{barcode} to be a multiset of intervals in $\R$, i.e., a multiset whose underlying set is a subset of $\I$.    

We say an $n$-module $M$ is \emph{pointwise finite dimensional}, or simply \emph{\pfd}, if $\dim(M_a) < \infty$ for all $a\in \R^n$.  

\subsubsection{Structure Theorem For Pointwise Finite Dimensional $1$-Modules}\label{Sec:WCBStructureTheorem}
The structure theorem for finitely generated $\Z$-indexed persistence modules \cite{zomorodian2005computing} is well known.  The recent paper of William Crawley-Boevey \cite{crawley2012decomposition} provides the following generalization:

\begin{thm}[Structure of Persistence Modules \cite{crawley2012decomposition}]\label{Thm:WCBStructureTheorem}
For any \pfd $1$-module $M$, there exists a unique barcode $\B_M$ such that \[M\cong \oplus_{I\in \B_M} C(I).\]
\end{thm}
As mentioned in \cref{Sec:Background}, we call $\B_M$ the \emph{barcode of $M$}.

\subsubsection{$\epsilon$-Matchings and the Bottleneck Distance}
To state the isometry theorem, we first need to define the bottleneck distance on \pfd $1$-modules.

For $I\subset \R$ an interval and $\epsilon\geq 0$, let the interval $\ex^\epsilon(I)$ be given by \[\ex^\epsilon(I)=\{t\in \R\mid \exists\, s\in I\textup{ with }|s-t|\leq \epsilon\}.\]  For $\D$ a barcode and $\epsilon\geq 0$, define $\D_{\epsilon}\subset \D$ to be the multiset of intervals in $\D$ which contain a subinterval of the form $[t,t+\epsilon]$ for some $t\in \R$.  Note that $\D_{0}=\D$.

Define an \emph{$\epsilon$-matching} between barcodes $\mathcal C$ and $\mathcal D$ to be a bijection $\sigma:\C'\leftrightarrow \D'$ for 
some $\C'\subset \C$, $\D'\subset \D$, satisfying the following properties:
\begin{enumerate*}
\item $\C_{2\epsilon} \subset \C'$,
\item $\D_{2\epsilon}\subset \D'$,
\item if 
$\sigma(I)=J$ then
$I\subset \ex^\epsilon(J)$ and 
$J\subset \ex^\epsilon(I)$.
\end{enumerate*}

For barcodes $\C$ and $\D$, we define the bottleneck distance $d_B$ by
\[d_B(\mathcal C,\mathcal D)=\inf\, \{\epsilon\in [0,\infty) \mid \exists\textup{ an }\epsilon\textup{-matching between }\mathcal C\textup{ and }\mathcal D\}.\]
$d_B$ induces a pseudometric on \pfd $1$-modules, also denoted $d_B$, given by
$d_B(M,N)=d_B(\B_M,\B_N)$.

\begin{remark}\label{delta_matching_remark}
Our definition of an $\epsilon$-matching is slightly stronger than the one appearing in \cite[Section 4.2]{chazal2012structure}, which is insensitive to whether intervals are closed or open on the left and right.  Using the stronger definition of $\epsilon$-matching allows us to state a sharp form of the isometry theorem for \pfd persistence modules.  However, regardless of which definition of $\epsilon$-matching one uses, the definition of the bottleneck distance one obtains is the same.   
\end{remark}

\subsubsection{The Algebraic Stability Theorem}
The \emph{algebraic stability theorem}, introduced in \cite{chazal2009proximity} and revisited in \cite{chazal2012structure}, considerably generalizes the earlier stability result for $\R$-valued functions of \cite{cohen2007stability}.  In its sharp formulation for \pfd 1-modules \cite{bauer2014induced}, the statement of the theorem is as follows:

\begin{thm}[Algebraic Stability Theorem]\label{AlgebraicStability}  For $M$ and $N$ \pfd $1$-modules, an $\epsilon$-interleaving morphism $f:M\to N(\epsilon)$ induces an $\epsilon$-matching between $\B_M$ and $\B_N$.  In particular, \[d_B(M,N)\leq d_I(M,N).\]\end{thm}

\subsection{The Isometry Theorem}\label{Thm:Isometry}
We now come to the isometry theorem, which tells us that the converse of the algebraic stability theorem also holds.  Our exposition closely follows \cite[Appendix B]{bauer2014induced}, which is in turn adapted from an earlier version of this paper.
\begin{thm}[The Isometry Theorem]\label{InterleavingEqualsBottleneck} For any $\epsilon\geq 0$, \pfd $1$-modules $M$ and $N$ are $\epsilon$-interleaved if and only if there exists an $\epsilon$-matching between $\B_M$ and $\B_N$.  In particular, \[d_I(M,N)=d_B(M,N)\] \end{thm}

Our proof of \cref{InterleavingEqualsBottleneck} relies on the following easy lemma, whose proof we leave to the reader:

\begin{lem}\label{Lem:CyclicConverse} 
Let $\epsilon\geq 0$. 
\begin{itemize}
\item[(i)] If $I$, $J$ are intervals such that $I\subset \ex^\epsilon(J)$ and $J\subset\ex^\epsilon(I)$,
then $C(I)$ and $C(J)$ are $\epsilon$-interleaved.
\item[(ii)] If $I$ is an interval which does not contain the subinterval $[t,t+2\epsilon]$ for any $t\in \R$, then $C(I)$ and the trivial module are $\epsilon$-interleaved.
\end{itemize}
\end{lem}

\begin{proof}[Proof of \cref{InterleavingEqualsBottleneck}]
In view of the algebraic stability theorem, it suffices to show that an $\epsilon$-matching between $\B_M$ and $\B_N$ induces an $\epsilon$-interleaving between $M$ and $N$.

We may assume without loss of generality that 
\[M=\bigoplus_{I\in \B_M} C(I),\qquad N=\bigoplus_{I\in \B_N} C(I). \]
For $\D_M\subset \B_M$ and $\D_N\subset \B_N$, let $\sigma:\D_M\to \D_N$ be an $\epsilon$-matching between $\B_M$ and $\B_N$.  Let 
\[
\begin{aligned}
M_\bullet&=\bigoplus_{I\in \D_M} C(I),\\
M_\circ&=\bigoplus_{I\in \D_M^c} C(I), 
\end{aligned}
\qquad
\begin{aligned}
N_\bullet&=\bigoplus_{I\in \D_N} C(I),\\
N_\circ&=\bigoplus_{I\in \D_N^c} C(I),
\end{aligned}
\]
where $\D_M^c$ and $\D_N^c$ denote the complements of $\D_M$ and $\D_N$ in $\B_M$ and $\B_N$, respectively.

Clearly, $M=M_\bullet\oplus M_\circ$ and $N=N_\bullet\oplus N_\circ$.  By \cref{Lem:CyclicConverse}\,(i), for each pair $(I,J)\in \D_M\times  \D_N$ with $\sigma(I)=J$, we may choose a pair of $\epsilon$-interleaving morphisms 
\[f_I:C(I)\to C(J)(\epsilon),\qquad g_J:C(J)\to C(I)(\epsilon).\]
These morphisms induce a pair of $\epsilon$-interleaving morphisms
\[f_\bullet:M_\bullet\to N_\bullet(\epsilon),\qquad g_\bullet:N_\bullet\to M_\bullet(\epsilon).\]

Define a morphism $f:M\to N$ by taking the restriction of $f$ to $M_\bullet$ to be equal to $f_\bullet$ and taking the restriction of $f$ to $M_\circ$ to be the trivial morphism.  Symmetrically, define a morphism $g:N\to M$ by taking the restriction of $g$ to $N_\bullet$ to be equal to $g_\bullet$ and taking the restriction of $g$ to $N_\circ$ to be the trivial morphism.  By \cref{Lem:CyclicConverse}\,(ii), \[\translation M {2\epsilon} (M_\circ)=\translation N {2\epsilon} (N_\circ)=0.\]  From this fact and the fact that $f_\bullet$ and $g_\bullet$ are $\epsilon$-interleaving morphisms, it follows that $f$ and $g$ are $\epsilon$-interleaving morphisms as well.
\end{proof}

  \section{Characterization of the $\epsilon$-Interleaving Relation}\label{Sec:Characterization}
We now present our characterization of the $\epsilon$-interleaving relation on $n$-modules; as noted in the introduction, this result expresses transparently the sense in which two $\epsilon$-interleaved $n$-modules are algebraically similar.  This characterization induces in an obvious way a corresponding characterization of $d_I$.  It is also the most important step in our proof of our main universality result Corollary~\ref{CorOptimality}.  
As noted earlier, our characterization of the $\epsilon$-interleaving relation holds, with essentially the same proof, for more general types of interleavings between $n$-modules \cite{lesnick2012multidimensional}.

The intuitive idea of our characterization theorem is simple: informally, the theorem tells us that two $n$-modules $M$ and $N$ are $\epsilon$-interleaved if and only if there exist presentations for $M$ and $N$ which differ from one another by $\vec \epsilon$-shifts of the grades of generators and relations.

\subsection{Free $n$-Modules and Presentations}\label{Sec:Presentations}
We begin our account of the characterization theorem by introducing free $n$-modules and presentations of $n$-modules.

\subsubsection{$n$-Graded Sets}
Define an {\it $n$-graded set} to be a pair \[\W=(W,\gr_\W)\] for some set $W$ and function $\gr_\W:W\to\R^n$.  We'll often abuse notation slightly and write $\W$ to mean the set $W$.  Also, when $\W$ is clear from context we'll write $\gr_\W$ simply as $\gr$.  Formally, we may regard $\W$ as the set of pairs \[\{(w,\gr(w))\mid w\in W\}.\]  We'll sometimes make use of this representation.   

For $\epsilon\geq 0$ let $\W(\epsilon)$ be the $n$-graded set $(W,\gr')$, where $\gr'(w)=\gr_\W(w)-\vec\epsilon$.  The disjoint union $\W_1 \amalg \W_2$ of $n$-graded sets $\W_1$ $\W_2$ is defined in the obvious way.  

Clearly, we may regard the set of homogeneous elements of an $n$-module $M$ as an $n$-graded set, so $\gr(y)$ is well defined for $y\in M$ homogeneous.

\subsubsection{Free $n$-Modules}
As we now explain, the usual notion of a free module extends to the setting of $n$-modules.

We regard the ring $P_n$ of \cref{MonoidRings} as an $n$-module by taking $(P_n)_a$ to be the 1-D vector space spanned by the monomial $x^a$, for each $a\in \R^n$.  Then, in our notion for shifts of persistence modules of \cref{ShiftsOfModules}, for $v\in \R^n$,  $P_n(-v)$ denotes a copy of the ring $P_n$, shifted so that the multiplicative identity of the ring is homogeneous of grade $v$.

For $\W$ an $n$-graded set, let $\free[\W]=\oplus_{w\in {\W}} P_n(-\gr(w))$.  We identify $\W$ with a homogeneous set of generators for $\free[\W]$ by identifying $w\in \W$ with the multiplicative identity of the corresponding summand $P_n(-\gr(w))$.  

A {\it free $n$-module} $F$ is an $n$-module such that $F\cong \free[\W]$ for some $n$-graded set $\W$.  Equivalently, we can define a free $n$-module as an $n$-module which satisfies a certain universal property; see \cite{carlsson2009theory} for the definition in the case of free multi-graded $k[x_1,\ldots,x_n]$-modules.  The definition in our case is analogous.

For $\Y$ a homogeneous subset of a free $n$-module $F$, let $\langle \Y \rangle$ denote the submodule of $F$ generated by $\Y$.

\subsubsection{Presentations of $n$-Modules}
A {\it presentation} of an $n$-module $M$ is a pair $(\W,\Y)$ where $\W$ is an $n$-graded set and $\Y \subset \free[\W]$ is a set of homogeneous elements such that $M\cong \free[\W]/\langle \Y \rangle$.  We denote the presentation $(\W,\Y)$ as $\langle \W|\Y \rangle$, and write $M\cong \langle \W|\Y \rangle$.

For $n$-graded sets $\W_1,\W_2$ and homogeneous sets $\Y_1,\Y_2\subset \free[\W_1 \amalg \W_2]$, we'll let $\langle \W_1,\W_2|\Y_1,\Y_2 \rangle$ denote $\langle \W_1 \amalg \W_2|\Y_1 \cup \Y_2 \rangle$.

Clearly, a presentation exists for any $n$-module.  If $M$ is an $n$-module such that there exists a presentation $\langle \W|\Y \rangle\cong M$ with $\W$ and $\Y$ finite, then we say that $M$ is {\it finitely presented}.

\subsubsection{Free Covers and Lifts}
Define a {\it free cover} of an $n$-module $M$ to be a surjective morphism $\rho_M:F_M \to M$, for $F_M$ a free $n$-module.  For $\rho_M:F_M\to M$, $\rho_N:F_N\to N$ free covers and $f:M \to N$ a morphism, define a {\it lift} of $f$ to be a morphism ${\tilde f}:F_M \to F_N$ such that the following diagram commutes:
 \[
\xymatrixcolsep{3.5pc}
\xymatrix{
F_M\ar[r]^{\tilde f}\ar[d]_{\rho_M} &F_N\ar[d]^{\rho_N}\\
M\ar[r]_{f} &N
}
\]

\begin{lem}[Existence and Uniqueness (up to Homotopy) of Lifts]\label{ExistenceAndHomotopyUniquenessOfLifts} For any morphism $f:M \to N$ of $n$-modules and free covers $\rho_M:F_M\to M$, $\rho_N:F_N\to N$, there exists a lift ${\tilde f}:F_M\to F_N$ of $f$.  If ${\tilde f'}:F_M\to F_N$ is another lift of $f$, then $\im({\tilde f}-{\tilde f'})\subset \ker \rho_N$.  \end{lem}

\begin{proof} This is just a specialization of the standard result on the existence and homotopy uniqueness of free resolutions \cite[Theorem A3.13]{eisenbud1995commutative} to the $0^{\th}$ modules in free resolutions for $M$ and $N$.  The proof is straightforward. \end{proof} 

\subsubsection{Characterization Theorem Preliminiares}  
To prepare for the statement of our characterization theorem, we make some basic observations and establish notation concerning shifts of $n$-graded sets and free $n$-modules.

\begin{remark}\label{ShiftInclusionRemark} \mbox{}
\begin{enumerate}[(i)]
\item For any $n$-graded set $\W$ and $\epsilon\in \R$, $\free[\W(\epsilon)]$ is canonically isomorphic to $\free[\W](\epsilon)$.  Thus we may identify $\free[\W(-\epsilon)](\epsilon)$ with $\free[\W](-\epsilon)(\epsilon)=\free[\W]$.
\item For $\epsilon\geq 0$, the morphism \[\varphi_{\free[\W(-\epsilon)]}^\epsilon:\free[\W(-\epsilon)] \to \free[\W]\] is injective, and so gives an identification of $\free[\W(-\epsilon)]$ with a submodule of $\free[\W]$. 
\item More generally, noting that for $n$-graded sets $\W_1$, $\W_2$, 
\[\free[\W_1,\W_2]=\free[\W_1]\oplus \free[\W_2],\] we see that the morphism 
\[
\left(
\begin{matrix}
\id_{\free[\W_1]}   &0 \\
0                             &\varphi_{\free[\W(-\epsilon)]}^\epsilon   
\end{matrix}
\right)
: \free[\W_1]\oplus\free[\W_2(-\epsilon)]\to \free[\W_1]\oplus\free[\W_2]
\]
gives an identification of $\free[\W_1,\W_2(-\epsilon)]$ with a submodule of $\free[\W_1,\W_2]$.  

Symmetrically, we obtain an identification of $\free[\W_1(-\epsilon),\W_2]$ with a submodule of $\free[\W_1,\W_2]$.
\end{enumerate}
\end{remark}

For $M$ an $n$-module, $\Y\subset M$ homogeneous, and $\epsilon\in \R$, let $\Y(\epsilon)\subset M(\epsilon)$ denote the image of $\Y$ under the bijection between $M$ and $M(\epsilon)$ induced by the identification of each summand $M(\epsilon)_a$ with $M_{a+\vec\epsilon}$.

\begin{remark}\label{Rmk:ShiftsAndHomogeneousSubsets}
By \cref{ShiftInclusionRemark}\,(i), for $\epsilon\geq 0$ and $\Y_1\subset \free[\W_1,\W_2(-\epsilon)]$ homogeneous, we may regard $\Y_1(-\epsilon)$ as a subset of $\free[\W_1(-\epsilon),\W_2(-2\epsilon)]$.  By \cref{ShiftInclusionRemark}\,(iii) then, we may identify $\Y_1(-\epsilon)$ with a homogeneous subset of $\free[\W_1(-\epsilon),\W_2]$.

Symmetrically, for $\Y_2\subset \free[\W_1(-\epsilon),\W_2]$ homogeneous, we may identify $\Y_2(-\epsilon)$ with a homogeneous subset of 
$\free[\W_1,\W_2(-\epsilon)]$.
\end{remark}
  
\subsection{The Characterization Theorem}
We now come to the main result of this section, our characterization of $\epsilon$-interleaved pairs of multidimensional persistence modules.  
  
\begin{thm}[Characterization Theorem]\label{AlgebraicRealization} $n$-modules $M$ and $N$ are $\epsilon$-interleaved if and only if there exist $n$-graded sets $\W_1,\W_2$ and homogeneous sets 
$\Y_1\subset \free[\W_1,\W_2(-\epsilon)]$, $\Y_2\subset \free[\W_1(-\epsilon),\W_2]$ such that
\begin{align*}
M\cong\langle\W_1,\W_2(-\epsilon)|\Y_1,\Y_2(-\epsilon)\rangle, \\
N\cong\langle\W_1(-\epsilon),\W_2|\Y_1(-\epsilon),\Y_2\rangle. \end{align*}
If $M$ and $N$ are finitely presented, $\W_1,\W_2,\Y_1,\Y_2$ can be taken to be finite. \end{thm}

\begin{remark}
For $M$ and $N$ $\epsilon$-interleaved $n$-modules, our proof of the characterization theorem in fact furnishes an explicit construction of  presentations of $M$ and $N$ as in the statement of the theorem.
\end{remark}

\begin{proof}[Proof of the characterization theorem]
It's easy to see that if there exist $n$-graded sets $\W_1,\W_2$ and sets $\Y_1,\Y_2$ as in the statement of the theorem, then $M$ and $N$ are $\epsilon$-interleaved.    

To prove the converse, we lift to free covers of $M$ and $N$ a construction presented in the proof of \cite[Lemma 4.6]{chazal2009proximity}.  \cite[Lemma 4.6]{chazal2009proximity} was stated only for $1$-modules, but the result and its proof generalize immediately to $n$-modules.  

To keep notation simple, throughout the proof we'll write the $\epsilon$-shift $f(\epsilon)$ of a morphism $f$ of $n$-modules simply as $f$.

Let $f:M\to N(\epsilon)$, $g:N\to M(\epsilon)$ be $\epsilon$-interleaving morphisms.  Upon generalizing to $n$-modules, the proof of \cite[Lemma 4.6]{chazal2009proximity} yields the following result as a special case:

\begin{lem}\label{ChazalInterpolationLemma}  
Let $\kappa^1:M(-2\epsilon)\to M\oplus N(-\epsilon)$ be given by \[\kappa^1(y)=(\varphi_{M(-2\epsilon)}^{2\epsilon}(y),-f(y)).\]  
Let $\kappa^2:N(-\epsilon)\to  M\oplus N(-\epsilon)$ be given by \[\kappa^2(y)=(-g(y),y).\]  
Let $R\subset M\oplus N(-\epsilon)$ be the submodule generated by $\im \kappa^1 \cup \im \kappa^2$.
Then \[M\cong (M\oplus N(-\epsilon))/R.\]   
\end{lem}
For convenience's sake, we reprove Lemma~\ref{ChazalInterpolationLemma} here. 
\begin{proof}    
Let \[\iota: M\to M\oplus N(-\epsilon)\] denote the inclusion, and let \[\zeta:M\oplus N(-\epsilon) \to  M\oplus N(-\epsilon)/R\] denote the quotient.  We'll show that $\zeta\circ \iota$ is an isomorphism.  For any $(y,z)\in M\oplus N(-\epsilon)$, $(-g(z),z)\in R$, so $\zeta\circ \iota\circ g(z)=(0,z)+R$.  Therefore $\zeta\circ\iota(g(z)+y)=(y,z)+R.$  Hence $\zeta\circ \iota$ is surjective.  

$\zeta\circ \iota$ is injective iff $\iota(M)\cap R=0$.  It's clear that $M\cap \im \kappa^2=0$.  Thus to show that $\zeta\circ \iota$ is injective it's enough to show that $\im \kappa^1\subset \im \kappa^2$.  If $y\in M(-2\epsilon)$, then since $\varphi_{M(-2\epsilon)}^{2\epsilon}(y)=g\circ f(y)$, \[\kappa^1(y)=(\varphi_{M(-2\epsilon)}^{2\epsilon}(y),-f(y))=(g\circ f(y),-f(y))=\kappa^2\circ f(-y).\]  Thus $\im \kappa^1\subset \im \kappa^2$ and so $\zeta\circ \iota$ is injective.

We conclude that $\zeta\circ \iota$ is an isomorphism.
\end{proof}

Now let $\langle \W_M|\Y_M\rangle$ be a presentation for $M$ and let $\langle \W_N|\Y_N\rangle$ be a presentation for $N$.  Without loss of generality we may assume that $M=\free[\W_M]/\langle \Y_M\rangle$ and $N=\free[\W_N]/\langle \Y_N\rangle$.  Let $\rho_M:\free[\W_M] \to M$, $\rho_N:\free[\W_N]\to N$ denote the quotient maps.  $\rho_M$ and $\rho_N$ are free covers for $M$ and $N$.  

Let ${\tilde f}:\free[\W_M]\to \free[\W_N(\epsilon)]$ be a lift of $f$ and let ${\tilde g}:\free[\W_N]\to \free[\W_M(\epsilon)]$ be a lift of $g$.
Let 
\begin{align*}
\Y_{M,N}&=\{y-{\tilde f}(y)\}_{y\in \W_M(-\epsilon)},\\
\Y_{N,M}&=\{y-{\tilde g}(y)\}_{y\in \W_N(-\epsilon)}.
\end{align*}  
Note that $\Y_{M,N}$ is a homogeneous subset of $\free[\W_M(-\epsilon),\W_N]$ and $\Y_{N,M}$ is a homogeneous subset of $\free[\W_M,\W_N(-\epsilon)]$.    

Let 
\begin{align*}
P_M&=\langle \W_M,\W_N(-\epsilon)|\Y_M,\Y_{N,M},\Y_N(-\epsilon),\Y_{M,N}(-\epsilon)\rangle\\
P_N&=\langle \W_M(-\epsilon),\W_N|\Y_M(-\epsilon),\Y_{N,M}(-\epsilon),\Y_N,\Y_{M,N}\rangle.
\end{align*}
By Remark~\ref{Rmk:ShiftsAndHomogeneousSubsets}, we can regard $\Y_{M,N}(-\epsilon)$ as a homogeneous subset of $\free[\W_M,\W_N(-\epsilon)]$, so $P_M$ is well defined.  Symmetrically, $P_N$ is well defined.  
 
We claim that $P_M$ is a presentation for $M$ and $P_N$ is a presentation for $N$.  We'll prove that $P_M$ is a presentation for $M$; the proof that $P_N$ is a presentation for $N$ is identical.  

Let \begin{align*}
F&=\free[\W_M,\W_N(-\epsilon)],\\
K&=\langle \Y_M,\Y_{N,M},\Y_N(-\epsilon),\Y_{M,N}(-\epsilon)\rangle \\
K'&=\langle \Y_M,\Y_N(-\epsilon)\rangle.
\end{align*}

Let $\rho:F\to F/K'$ denote the quotient map.  Clearly, we may identify $F/K'$ with $M\oplus N(-\epsilon)$ and $\rho$ with the map
\[
\left(
\begin{matrix}
\rho_M   &0 \\
0                             &\rho_N  
\end{matrix}
\right)
: \free[\W_1]\oplus\free[\W_2(-\epsilon)]\to M\oplus N(-\epsilon).
\]
We'll check that $\rho$ maps $\langle \Y_{M,N}(-\epsilon) \rangle$ surjectively to $\im \kappa^1$ and $\langle \Y_{N,M} \rangle$ surjectively to $\im \kappa^2$, so that under the identification of $F/K'$ with $M\oplus N(-\epsilon)$, $K/K'=R$.  Given this, it follows that $P_M$ is a presentation for $M$ by Lemma~\ref{ChazalInterpolationLemma} and the third isomorphism theorem for modules \cite{dummit1999abstract}.
\begin{sloppypar}
We first check that $\langle \rho(\Y_{M,N}(-\epsilon)) \rangle = \im\kappa^1$.  Viewing $\Y_{M,N}(-\epsilon)$ as a subset of $\free[\W_M, \W_N(-\epsilon)]$, \[\Y_{M,N}(-\epsilon)=\{\varphi_{\free[\W_M(-2\epsilon)]}^{2\epsilon}(y)-{\tilde f}(y)\}_{y\in \W_M(-2\epsilon)}.\]  $\varphi_{\free[\W_M(-2\epsilon)]}^{2\epsilon}$ is a lift of $\varphi_{M(-2\epsilon)}^{2\epsilon}$ and ${\tilde f}$ is a lift of $f$, so for any $y\in \W_M(-2\epsilon)$, \[\rho(\varphi_{\free[\W_M(-2\epsilon)]}^{2\epsilon}(y)-{\tilde f}(y))=(\varphi_{M(-2\epsilon)}^{2\epsilon}\circ\rho_M(y),-f\circ \rho_M(y))= \kappa^1\circ \rho_M(y).\]  Thus $\rho(\Y_{M,N}(-\epsilon))\subset \im \kappa^1$.  Since $\W_M$ generates $\free[\W_M]$ and $\rho_M$ is surjective, we have that $\rho\langle \Y_{M,N}(-\epsilon) \rangle =\im \kappa^1$.

The check that $\rho\langle \Y_{N,M} \rangle = \im \kappa^2$ is similar to the above, but simpler.  $\Y_{N,M}=\{y-{\tilde g}(y)\}_{y\in \W_N(-\epsilon)}$.  ${\tilde g}$ is a lift of $g$ so for any $y\in \W_N(-\epsilon)$, \[\rho(y-{\tilde g}(y))=(-g\circ \rho_N(y),\rho_N(y))=\kappa^2\circ \rho_N(y).\]  Thus $\rho(\Y_{N,M})\subset \im\kappa^2$.  Since $\W_N$ generates $\free[\W_N]$ and $\rho_N$ is surjective, we have that $ \rho\langle \Y_{N,M}\rangle=\im \kappa^2$.   
\end{sloppypar}
This completes the verification that $P_M$ is a presentation for $M$.  

Now, taking \[\W_1=\W_M,\quad \W_2=\W_N,\quad \Y_1=\Y_M \cup \Y_{N,M},\quad \Y_2=\Y_N \cup \Y_{M,N}\] gives the first statement of Theorem~\ref{AlgebraicRealization}.  If $M$ and $N$ are finitely presented then $\W_M, \W_N, \Y_M, \Y_N, \Y_{M,N}$, and $\Y_{N,M}$ can all be taken to be finite; the second statement of Theorem~\ref{AlgebraicRealization} follows. \end{proof} 

\begin{example}\label{Ex:CharacterizationExample} We now present an explicit example of the compatible presentations of $\epsilon$-interleaved $n$-modules constructed in the proof of the characterization theorem~\ref{AlgebraicRealization}.
Let $M$ and $N$ be $1$-modules given by 
\begin{align*}
M&=\langle (a,0)|x^3 a\rangle,\\
N&=\langle (b,1)|x^2 b\rangle.
\end{align*}

$M$ and $N$ are $1$-interleaved: Let $\tilde f:\free[(a,0)]\to \free[(b,0)]$ be the morphism which sends $a$ to $b$ and let $\tilde g:\free[(b,1)] \to \free[(a,-1)]$ be the morphism which sends $b$ to $x^2 a$.  $\tilde f$ and $\tilde g$ descend to $1$-interleaving morphisms $f:M\to N(1)$, $g:N\to M(1)$, so that  $\tilde f$ is a lift of $f$ and $\tilde g$ is a lift of $g$.  

Let 
\[
\W_1=\{(a,0)\},\quad \W_2=\{(b,1)\},
\]
and define $\Y_1, \Y_2\subset  \free[\W_1, \W_2]$ by 
\[
\Y_1=\{x^3 a,x b-x^2 a\}, \quad \Y_2=\{x^2 b,xa-b\}.
\]
The construction in the proof of the characterization theorem gives us that 
\begin{align*}
M &\cong \langle\W_1,\W_2(-\epsilon)|\Y_1,\Y_2(-\epsilon)\rangle, \\
N&\cong \langle\W_1(-\epsilon),\W_2|\Y_1(-\epsilon),\Y_2\rangle. 
\end{align*} 
Noting that $\free[\W_1,\W_2(-\epsilon)]$ and $\free[\W_1(-\epsilon),\W_2]$ are canonically isomorphic to the free submodules $\langle a, xb\rangle$ and $\langle xa, b\rangle$  of  $\free[(a,0),(b,1)]$, we thus have that
\begin{align*}
M&\cong \langle a, xb\rangle/  \langle  x^3 a,xb-x^2 a,x^3 b,x^2 a- xb \rangle, \\
N&\cong  \langle xa, b\rangle/  \langle  x^4 a,x^2 b-x^3 a,x^2 b,xa- b  \rangle. 
\end{align*} 
Equivalently, but more intrinsically, we may write
\begin{align*}
M&\cong \langle (a,0),(b,2)|x^3 a,b-x^2 a,x^2 b,x^2 a- b\rangle\\
N&\cong \langle (a,1),(b,1)|x^3 a,x^2 b-x^2 a,x^2 b,a- b\rangle.
\end{align*}
\end{example}

  \section{Universality of the Interleaving and Bottleneck Distances}\label{Sec:OptimalitySection}
This section presents our universality results for $d_I$ and $d_B$.  We also consider the stability properties of $d_I$.  

\subsubsection{A Category of $\R^n$-valued functions}\label{Sec:SublevelsetFiltrations}
For $n\geq 1$, let $\CS$ be the category defined as follows:
\begin{enumerate*}
\item Objects of $\CS$ are functions $\gamma:X\to \R^n$, for $X$ any topological space.  
\item For functions $\gamma^X:X\to \R^n$ and $\gamma^Y:Y\to \R^n$, $\hom(\gamma^X,\gamma^Y)$ is the set of continuous functions $f:X\to Y$ such that $\gamma^X(p)\geq \gamma^Y\circ f(p)$ for all $p \in X$.  
\end{enumerate*}

\subsubsection{Multidimensional Filtrations}\label{GeometricPreliminariesSection}

Recall that in \cref{Sec:Background}, we defined an $n$-dimensional filtration to be a functor $\F:\RCat^n\to \Top$ that maps each element of $\hom(\RCat^n)$ to inclusion.  Let $\nfilt$ denote the category whose objects are $n$-dimensional filtrations and whose morphisms are natural transformations.   

\subsubsection{The Sublevelset Filtration Functor}\label{Sec:SublevelsetFiltrations}
In \cref{Ex:SublevelsetFiltration}, we defined the $n$-dimensional sublevelset filtration $\FS(\gamma)$ of any function $\gamma:T\to \R^n$ with $T$ a topological space.  This gives a map 
\[\FS:\obj \CS\to \obj \nfilt.\]  In fact, defining the action of $\FS$ on morphisms in the obvious way, we obtain a functor \[\FS:\CS\to \nfilt.\]

\subsubsection{Multidimensional Persistent Homology}
As in \cref{Sec:Background}, for $i\geq 0$ let $H_i$ denote the $i^{\th}$ singular homology functor with coefficients in the field $k$.  $H_i$ induces a functor \[H_i:\nfilt\to \nmod,\] the  \emph{$i^{\th}$ persistent homology functor}.  When no confusion is likely, we'll often write the composition $H_i \FS:\CS\to\nmod$ simply as $H_i$.

\subsection{Stability of the Interleaving Distance}

\subsubsection{A Metric on $[\obj \CS ]$}
For $\gamma:X\to \R^n$ a function, let 
\begin{equation*}
\|\gamma\|_{\infty}=
\begin{cases}
\sup_{p\in X}\|\gamma(p)\|_{\infty} &\textup{if } X\ne \emptyset, \\
0 &\textup{if } X=\emptyset.
\end{cases}
\end{equation*} 
Given $\gamma^X:X\to \R^n ,\gamma^Y:Y\to \R^n$, we let \[d_\infty(\gamma^X,\gamma^Y)=\inf_{h\in {\mathcal H}} \|\gamma^X-\gamma^Y\circ h\|_{\infty},\]
where ${\mathcal H}$ is the set of homeomorphisms from $X$ to $Y$.  $d_\infty$ descends to a metric on $[\obj \CS]$, which we also write as $d_\infty$.  
Note that $d_\infty(\gamma^X,\gamma^Y)=\infty$ if $X$ and $Y$ are not homeomorphic.

When $n=1$, $d^S$ is known as the {\it natural pseudo-distance}; it features prominently in the work of Patrizio Frosini and his coauthors on persistent homology---see \cite{d2010natural}, for example.  

\begin{remark}\label{GeometricInterleavingRemark}
It is interesting to note that $d_\infty$ can be defined equivalently as an interleaving distance.  To define an interleaving distance on $[\obj \CS]$, we need to give a suitable definition of $\epsilon$-interleavings in the category $\CS$.  To do this, for each $\epsilon\geq 0$ we have to specify an $\epsilon$-shift functor $(\cdot)(\epsilon): \CS\to \CS$ and a transition morphism $\varphi_{\gamma}^\epsilon:\gamma\to \gamma(\epsilon)$ for every $\gamma\in \obj \CS$.  

For $\gamma:X\to \R^n$ in $\obj \CS$, we let $\gamma(\epsilon)=\gamma'$, where $\gamma':X\to \R^n$ is given by $\gamma'(x)=\gamma(x)-\vec\epsilon$; for $f$ a morphism in $\CS$, we let $f(\epsilon)=f$; and we let $\varphi_{\gamma}^\epsilon=\id_X$.

It's easy to check that the interleaving distance induced by these choices is equal to $d_\infty$.  
\end{remark}

\subsubsection{Stability of the Bottleneck Distance}
Here is the fundamental stability result for sublevelset persistent homology.  
\begin{thm}[Stability of $d_B$ \cite{cohen2007stability,chazal2009proximity}]\label{OrdinarySublevelsetStability}  For $i\geq 0$, topological spaces $X,Y$, and functions 
\[
\gamma^X:X\to \R, \quad
\gamma^Y:Y\to \R
\]
such that $H_i(\gamma^X)$ and $H_i(\gamma^Y)$ are \pfd, we have
\[d_B(H_i(\gamma^X),H_i(\gamma^Y))\leq d_\infty(\gamma^X,\gamma^Y).\]
\end{thm}
As shown in \cite{chazal2009proximity}, this result is an immediate corollary of the algebraic stability theorem~\ref{AlgebraicStability}.

  \subsubsection{Stability of the Interleaving Distance}\label{MultidimensionalStabilitySection}

We now make the easy observation that $d_I$ is stable with respect to multidimensional sublevelset persistent homology; in view of the isometry theorem (\cref{InterleavingEqualsBottleneck}) this generalizes Theorem~\ref{OrdinarySublevelsetStability}.  In fact, $d_I$ is stable with respect to multidimensional persistent homology in several other senses as well; see \cite{lesnick2012multidimensional} for further (easy) stability results.  

For $i\geq 0$, we say a pseudometric $d$ on $[\obj \nmod]$ is \emph{$i$-stable} if 
for any topological spaces $X,Y$ and functions 
\[\gamma^X:X\to \R^n, \quad \gamma^Y:Y\to \R^n,\] we have
\[d(H_i(\gamma^X),H_i(\gamma^Y))\leq d_\infty(\gamma^X,\gamma^Y).\]
We say a pseudometric on $[\obj \nmod]$ is \emph{stable} if it is $i$-stable for all $i\geq 0$.
\begin{thm}\label{MultidimensionalSublevelsetStability} $d_I$ is stable.  \end{thm}
\begin{proof}  
If $d_\infty(\gamma^X,\gamma^Y)=\epsilon$, then for any $\delta>\epsilon$, there exists a homeomorphism $h:X\to Y$ such that  for $a\in \R^n$, $\FS(\gamma^X)_a\subset \FS(\gamma^Y\circ h)_{a+\vec\delta}$ and $\FS(\gamma^Y\circ h)_a\subset \FS(\gamma^X)_{a+\vec\delta}$.  
The images of these inclusions under the $i^{th}$ singular homology functor define $\delta$-interleaving morphisms between $H_i(\gamma^X)$ and $H_i(\gamma^Y\circ h)$.  $\gamma^Y\circ h$ and $\gamma^Y$ are isomorphic objects of $\CS$, so $H_i(\gamma^Y\circ h)$ and $H_i(\gamma^Y)$ are isomorphic, i.e., $0$-interleaved.  Thus $H_i(\gamma^X)$ and $H_i(\gamma^Y)$ are $\delta$-interleaved.  It follows that \[d_I(H_i(\gamma^X),H_i(\gamma^Y))\leq \epsilon,\] as needed. 
\end{proof}

  \subsection{Universality of the Interleaving Distance}\label{OptimalityGeneralities}
We now are ready to formulate and prove our main universality results.

\subsubsection{Definitions of Universality}\label{Sec:OptimalityDefs}
For $i\geq 0$, we say that a pseudometric $d$ on $[\obj \nmod]$ is \emph{$i$-universal} if $d$ is $i$-stable and for any other $i$-stable metric $d'$ on $[\obj \nmod]$, $d'(M,N)\leq d(M,N)$ for all $M,N\in \im H_i \FS$.

We say a pseudometric $d$ on $[\obj \nmod]$ is \emph{universal} if $d$ is stable and for any other stable metric $d'$ on $[\obj \nmod]$, $d'(M,N)\leq d(M,N)$ for all $M,N$ such that there exists $i\geq 0$ with $M,N\in \im H_i \FS$.

Recall that $k$ is the field of coefficients with respect to which we have defined $\nmod$ and $H_i$.  When $k$ is a prime field, our definitions can be simplified:   

\begin{lem}
In the case that our field of coefficients $k$ is prime,
\begin{enumerate}[(i)]
\item for $i\geq 1$, a pseudometric $d$ on $[\obj \nmod]$ is $i$-universal if and only if $d$ is $i$-stable and for any other $i$-stable metric $d'$ on $[\obj \nmod]$, $d'\leq d$,
\item a pseudometric $d$ on $[\obj \nmod]$ is universal if and only if $d$ is stable and for any other stable metric $d'$ on $[\obj \nmod]$, $d'\leq d$.
\end{enumerate}
\end{lem} 
\begin{proof}
Proposition~\ref{RealizationProp} below implies that when $k$ is a prime field, $H_i$ is essentially surjective for $i\geq 1$.  Given this, the result is immediate.  
\end{proof}

\subsubsection{The Main Universality Result}
\begin{thm}\label{MainOptimality} If $k$ is a prime field and $i\geq 1$, then $d_I$ is $i$-universal.  
\end{thm}
We give the proof of \cref{MainOptimality} below.

\begin{cor}\label{CorOptimality}
If $k$ is a prime field then $d_I$ is universal.
\end{cor}
\begin{proof}[Proof of Corollary~\ref{CorOptimality}]
Theorem~\ref{MultidimensionalSublevelsetStability} shows that $d_I$ is stable.  Let $d$ be another stable metric.  Then $d$ is in particular 1-stable, so by Theorem~\ref{MainOptimality}, $d\leq d_I$.  
\end{proof}
In \cref{Sec:UniversalityOfBottleneck}, we also present an analogue of Theorems~\ref{MainOptimality} for \pfd $1$-modules which holds for arbitrary fields $k$ and $i\geq 0$.

I suspect that Theorem~\ref{MainOptimality} strengthens as follows:
\begin{conjecture}\label{Conj:ZeroOptimality}
For any field $k$ and $i\geq 0$, $d_I$ is $i$-universal.
\end{conjecture}
\subsubsection{Lifts of Interleavings to Functions}
The key step in the proof of Theorem~\ref{MainOptimality} is the proof of the following proposition.%
\begin{prop}[Existence of Geometric Lifts of Interleavings]\label{RealizationProp}Let $k$ be a prime field and let $M$ and $N$ be $\epsilon$-interleaved $n$-modules.  Then for any $i\geq 1$, there exists a CW-complex $X$ and continuous functions $\gamma^M,\gamma^N:X\to \R^n$ such that 
\[
M\cong H_i(\gamma^M),\quad N\cong H_i(\gamma^N), \quad d_\infty(\gamma^M,\gamma^N)=\epsilon.
\]
\end{prop}

The proposition tells us that interleavings on $n$-modules lift to interleavings on objects of $\CS$, in the sense of Remark~\ref{GeometricInterleavingRemark}.

Section~\ref{Sec:ProofGeoLifts} below is devoted to the proof of Proposition~\ref{RealizationProp}.  

\begin{proof}[Proof of Theorem~\ref{MainOptimality}] 
We now deduce Theorem~\ref{MainOptimality} from Proposition~\ref{RealizationProp}.  Let $M$ and $N$ be $n$-modules such that $d_I(M,N)=\epsilon$.  For any $\delta>0$, $M$ and $N$ are $(\epsilon+\delta)$-interleaved.  By Proposition~\ref{RealizationProp}, for $i\geq 1$ there exists a topological space $X$ and $\gamma^M:X\to \R^n$, $\gamma^N:X\to \R^n$ such that $M\cong H_i(\gamma^M)$, $N\cong H_i(\gamma^N)$, and $d_\infty(\gamma^M,\gamma^N)=\epsilon+\delta$.  Thus if $d$ is any $i$-stable metric on $[\obj \nmod]$, $d(M,N)\leq \epsilon+\delta.$  Since this holds for all $\delta>0$, we have $d(M,N)\leq \epsilon=d_I(M,N)$.   
\end{proof}

Note that Conjecture~\ref{Conj:ZeroOptimality} would follow from the following conjectural extension of Proposition~\ref{RealizationProp}. 

\begin{conjecture} Let $k$ be any field and let $M$ and $N$ be $\epsilon$-interleaved $n$-modules in $\im H_i\FS$.  Then for any $i\geq 0$, there exists a CW-complex $X$ and continuous functions $\gamma^M,\gamma^N:X\to \R^n$ such that 
\[M\cong H_i(\gamma^M), \quad N\cong H_i(\gamma^N), \quad d_\infty(\gamma^M,\gamma^N)\leq\epsilon.\]
\end{conjecture}

\subsection{Proof of Proposition~\ref{RealizationProp}. }\label{Sec:ProofGeoLifts}

\subsubsection{Part 1: Constructing the CW-complex}\label{ConstructingComplexSection}
It is clear that Proposition~\ref{RealizationProp} is true if $M$ and $N$ are both trivial $n$-modules, so we may assume without loss of generality that either $M$ or $N$ is not trivial.

Theorem~\ref{AlgebraicRealization} gives us $n$-graded sets $\W_1,\W_2$ and homogeneous sets $\Y_1,\Y_2 \subset \free[\W_1,\W_2]$ such that  $\Y_1\in \free[\W_1,\W_2(-\epsilon)]$, $\Y_2\in \free[\W_1(-\epsilon),\W_2]$, and
\begin{align*} 
M &\cong \langle\W_1,\W_2(-\epsilon)|\Y_1,\Y_2(-\epsilon)\rangle,\\ 
N&\cong \langle\W_1(-\epsilon),\W_2|\Y_1(-\epsilon),\Y_2\rangle. 
\end{align*}

Given such $\W_1,\W_2,\Y_1,\Y_2$, we now construct the CW-complex $X$ appearing in the statement of Proposition~\ref{RealizationProp}.  

Let $\W=\W_1 \amalg \W_2$ and $\Y=\Y_1 \cup \Y_2$.  Let $X'$ be the standard CW-complex structure on $\R$.  That is, for each $z\in \Z$, we take  $z$ to be a $0$-cell in $X'$ and we take the interval $(z,z+1)$ to be a $1$-cell in $X'$.  

Now fix $i\geq 1$.  For $Z$ a CW-complex, let $\cells(Z)$ the collection of cells of $Z$.  We define $X$ so that

\begin{enumerate*}
\item $X'$ is a subcomplex of $X$.
\item $X$ has an $i$-cell $e^i_w$ for each $w\in \W$. 
\item $X$ has an $(i+1)$-cell $e^{i+1}_y$ for each $y\in \Y$.
\item $\cells(X)=\cells(X')\amalg \{e^i_w\}_{w\in\W} \amalg \{e^{i+1}_y\}_{y\in \Y}.$
\end{enumerate*}

For $a=(a_1,\ldots,a_n)\in\R^n$, let \[\lfloor a \rfloor=\max\{z\in \Z \mid z\leq a_j\textup{ for }1\leq j\leq n\}.\]  For all $w\in \W$, let the attaching map of $e^i_w$ be the constant map to the $0$-cell $\lfloor \gr(w) \rfloor \in X'$. 

This defines the $i$-skeleton $X^i$ of $X$.  $X^i$ is thus a copy of the real line with a copy $S_w^i$ of the $i$-dimensional sphere attached for each $w\in \W$.  

Clearly, the map \[q:X^i \to \vee_{w\in \W}\, S_w^i\] which collapses $X'$ to a point is a homotopy equivalence, so the map \[q_*:  \pi_i(X^i)\to \pi_i(\vee_{w\in \W}\, S^i_w)\] is an isomorphism.  By \cite[Examples 4.26 and 1.21]{hatcher2002algebraic}, for $i>1$ (i=1), $\pi_i(\vee_{w\in \W}\, S^i_w)$ is free abelian (free) with generators the homotopy classes of the inclusions \[j_w:S_w^i\hookrightarrow \vee_{w\in \W}\, S^i_w.\]  The image of this set of generators under $q_*^{-1}$ is thus a generating set for $\pi_i(X^i)$.  We let $\Gen_w\in \pi_i(X^i)$ denote the generator $q_*^{-1}[j_w]$.  

To complete the construction of $X$, it remains only to specify the attaching map $\sigma_y:S^i \to X^i$ of each cell $e^{i+1}_y$.  Here we need to treat the cases $k=\Q$ and $k=\Z/p\Z$ separately.  

We first consider the case $k=\Q$.  Recall that we identify $\W$ with a set of homogeneous generators for $\free[\W]$, and that by definition, $\Y\subset \free[\W]$.  For each $y\in \Y$, there exists a unique choice of finite set $\W_y\subset\W$ and non-zero rational number $c'(y,w)$ for each $w\in \W_y$, such that $\gr(w)\leq \gr(y)$ and 
\begin{align*}
y=\sum_{w\in \W_y} c'(y,w) x^{\gr(y)-\gr(w)}w.
\end{align*}
Since $\W_y$ is finite, there exists $z\in \Z$, $z\ne 0$, such that for each $w\in \W_y$, $c'(y,w)z\in \Z$.  For each $w \in\W_y$, we let $c(w,y)=c'(w,y)z$.

Analogously, in the case that $k=\Z/p\Z$, for each $y\in \Y$ there exists a unique choice of finite set $\W_y\subset \W$ and non-zero integer $c(y,w)$ for each $w\in \W_y$, such that $\gr(w)\leq \gr(y)$ and 
\begin{align*}
y=\sum_{w\in \W_y} [c(y,w)] x^{\gr(y)-\gr(w)}w,
\end{align*}
where $[c(y,w)]\in \Z/p\Z$ denotes the equivalence class of $c(y,w)\, \mathrm{ mod }\ p$.

Having defined the integers $c(y,w)$ differently in the two cases, the rest of the proof of Proposition~\ref{RealizationProp} is the same for both cases.

We define the attaching map $\sigma_y: S^i \to X^i$ of the cell $e^{i+1}_y$ to be any map such that
\begin{enumerate*}
\item $\im \sigma_y \subset X'\, \bigcup_{w\in \W_y} {e^i_w},$
\item  regarding $X'$ as a copy of the real line, we have that for each $r\in  \im \sigma_y \cap X'$, \[r\leq \max_{w\in \W_y} \lfloor \gr(w) \rfloor,\]
\item $\sigma_y$ is in the unbased homotopy class containing the based homotopy class 
\[\prod_{w\in \W_y }\Gen_{w}^{c(y,w)} \in \pi_i(X^i).\]
\end{enumerate*}
It is easy to check that such a map $\sigma_y:S^i\to X^i$ exists.  

\begin{remark}
In the special case that the grades of elements of $\W$ are bounded below in the partial order on $\R^n$ (for example, when $\W$ is finite), it suffices to work with a simpler definition of $X$, where $X^i$ is taken to be a wedge sum of $i$-spheres.  However, in the general case, this simpler construction of $X$ does not suffice.
\end{remark}

\begin{example}\label{Ex:CWConstruction}
We illustrate the construction of the CW-complex $X$ above with a simple example.  Suppose that $n=2$, $i=1$, and 
\begin{align*}
\W&=\{(a,(1/2,1)),(b,(3,2)),(c,(5,5))\},\\
\Y&=\{x_1x_2a,c-x_1^2x_2^3b\}.
\end{align*}
We obtain $X^1$ by attaching three 1-spheres $S^1_a$, $S^1_b$, and $S^1_c$ to $X'$ at 0, 2, and 5, respectively.  We obtain $X$ from $X^1$ by attaching a disk along $S^1_a$ and a second disk along 
\[\ell\cdot (S^1_b)^{-1} \cdot \ell^{-1} \cdot S^1_c,\]
where $\ell:[0,1]\to X'=\R$ is the linear path from 2 to 5, and we interpret $S^1_b$, $S^1_c$ as closed paths in $X^1$ with endpoints in $X'$.  The resulting space $X$ is a copy of $X'$ with a disk attached to $0\in X'$ and a cylinder attached to $[2,5]\subset X'$, as in Figure 1.
\end{example}
\begin{figure}[hbt]
\label{Fig:CW_Construction}
\centering
\includegraphics[width=.85\textwidth]{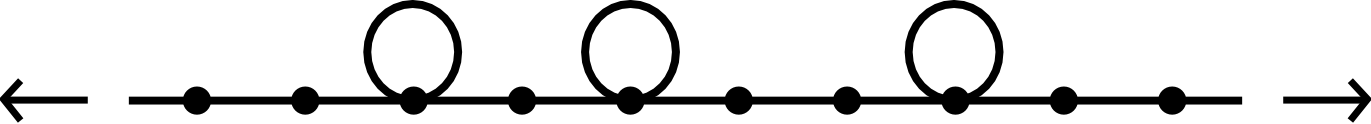}\\
\vskip30pt
\includegraphics[width=.85\textwidth]{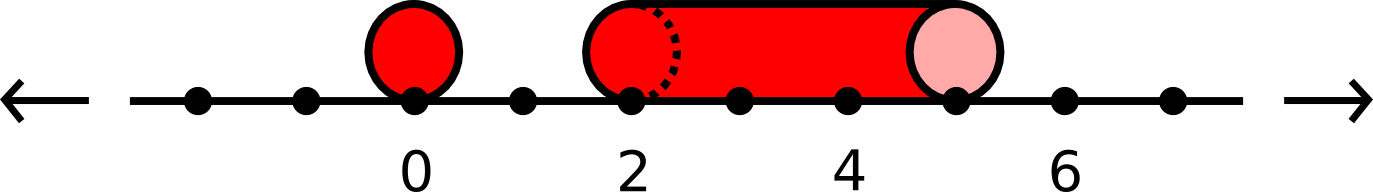}
{\caption{The CW complexes  $X^1\subset X$ constructed in \cref{Ex:CWConstruction}.} \label{Fig:induced-matching-close-eps}}
\end{figure}

For $T$ a CW-complex and $l\geq 0$, let $C_l(T)$ denote the $l^{\rm{th}}$ cellular chain vector space of $T$, and let \[\partial^T_{l}:C_l(T)\to C_{l-1}(T)\] denote the $l^{\rm{th}}$ cellular boundary map.  (See \cite{hatcher2002algebraic} for details on cellular homology.)  

Recall that the $l$-cells in $T$ form a basis for $C_l(T)$.  

\begin{lem}\label{Lem:DegreeLemma}
For $y\in \Y$, \[\partial^X_{i+1}(e^{i+1}_y)=\sum_{w\in \W_y} c(y,w) e^i_w,\]
where we interpret the equation mod $p$ if $k=\Z/p\Z$.  
\end{lem}

\begin{proof}
For $e$ an $i$-cell in $X$, let $q_e:X\to S^i$ denote the map which collapses the complement of $e$ to a point.  We first note that by the cellular boundary formula \cite{hatcher2002algebraic},
\[\partial^X_{i+1}(e^{i+1}_y)=\sum_{\text{$e$ an $i$-cell in $X$}} d_e\, e\]
where $d_e$ is the degree of the map $q_e\circ \sigma_y:S^i\to S^i$, and only finitely many of the coefficients $d_e$ are non-zero.  As above, we interpret the equation mod $p$ if $k=\Z/p\Z$.

If $i=1$ and $e$ is a $1$-cell in $X'$, then $d_e=0$.  To see this, note that $q_e$ factors through the map $X^1\to X'$ which, for each $w\in \W$, collapses the 1-cell $e^1_w$ onto its point of intersection with $X'$.  Since $X'$ is contractible, $q_e$ is nullhomotopic, so $q_e\circ \sigma_y$ is nullhomotopic as well.  Therefore $d_e=0$.

For $i\geq 1$, if $e$ is an $i$-cell in $X$ with $\im \sigma_y\cap e=\emptyset$, then again we have $d_e=0$.  Since $\sigma_y \subset X'\,  \bigcup_{w\in \W_y} {e^i_w},$ we thus have that
\[\partial^X_{i+1}(e^{i+1}_y)=\sum_{w\in \W_y} d_w e^i_w,\]
where we have written $d_{e^{i}_w}$ simply as $d_w$.  

It remains to check that for each $w\in \W_y$, $d_w=c(y,w)$.  Let $q_w=q_{e^i_w}$, and let \[\bar q_w:\vee_{v\in \W}\, S^i_v\to S^i\] be the map which collapses $\vee_{v\in \W-\{w\}}\, S^i_v$ to a point.
We have a commutative diagram:
\[
\xymatrixcolsep{5.5pc}
\xymatrix{
X^i\ar[r]^{q}\ar[dr]_{q_w} &\vee_{v\in \W}\, S^i_v\ar[d]^{\bar q_w}\\
& S^i\\
}
\]
$\sigma_y$ lies in the unbased homotopy class containing the homotopy class \[\prod_{w\in \W_y }\Gen_{w}^{c(y,w)}\in \pi_i(X^i),\]  so by the way we've defined the generators $\Gen_{w}$, 
we have that $q \circ \sigma_y\simeq j$, where $j=\prod_{w\in \W_y } j_w^{c(y,w)}$.  
By the commutativity of the above diagram, then, we have \[q_w\circ \sigma_y\simeq\bar q_w\circ j\simeq (\id_{S^i})^{c(y,w)}.\]  By \cite[Corollary 4.25]{hatcher2002algebraic},  $(\id_{S^i})^{c(y,w)}$ is a map of degree $c(y,w)$, so since the degree of a map is a homotopy invariant, we have that $d_w=c(y,w)$, as desired.
\end{proof}

\subsubsection{Part 2: Defining $\gamma^M$ and $\gamma^N$}
Having defined the CW-complex $X$, we next define $\gamma^M,\gamma^N:X\to \R^n$.  Let 
\[{\tilde X}=X' \ \amalg_{w\in \W} D^i_w \ \amalg_{y\in \Y} D^{i+1}_y,\]
where $D^i_w$ and $D^{i+1}_y$ denote copies of the i-dimensional and  (i+1)-dimensional closed unit disk, respectively.  
$X$ is the quotient of $\tilde X$ under the equivalence relation generated by the attaching maps of the $i$-cells and $(i+1)$-cells of $X-X'$.  Let $\rho:\tilde X  \to X$ denote the quotient map.  Any continuous function $\tilde X\to \R^n$ which is constant on fibers of $\rho$ descends to a continuous function $X\to \R^n$.
To define $\gamma^M,\gamma^N$, we define functions ${\tilde \gamma^M},{\tilde \gamma^N}:\tilde X\to \R^n$ which are constant on fibers of $\rho$.  We then take $\gamma^M,\gamma^N$ to be the respective induced functions on $X$.

We'll take both ${\tilde \gamma^M}$ and ${\tilde \gamma^N}$ to have the property that for each disk in $\tilde X-X'$, the restriction of the function to any {\it radial line segment} (i.e., a line segment from the origin of the disk to the boundary of the disk) is linear.  To specify ${\tilde \gamma^M}$ and ${\tilde \gamma^N}$, then, it is enough to specify the values of each function on $X'$ and on the origin $O_D$ of each disk $D$ in $\tilde X-X'$.  We define ${\tilde \gamma^M}$ and ${\tilde \gamma^N}$ as follows:
\begin{align*}
\tilde\gamma^M(t)&=\tilde\gamma^N(t)=\vec t\ \ \, \textup{ for }t\in X'=\R,\\
{\tilde \gamma^M}(O_{D_v})&=
\begin{cases}
\gr(v)                      &\textup{ for } v\in \W_1 \cup \Y_1\\
\gr(v)+\vec \epsilon&\textup{ for } v\in \W_2 \cup \Y_2,
\end{cases}\\
{\tilde \gamma^N}(O_{D_v})&=
\begin{cases}
\gr(v)+\vec \epsilon&\textup{ for } v\in \W_1 \cup \Y_1\\
\gr(v)                      &\textup{ for } v\in \W_2 \cup \Y_2.
\end{cases}
\end{align*}

\begin{lem}\label{FunctionDistanceLemma}
$d_\infty(\gamma^M,\gamma^N)=\epsilon$.
\end{lem}

\begin{proof} Assume that for a disk $D$ of $\tilde X-X'$, $\|{\tilde \gamma^M}(p)-{\tilde \gamma^N}(p)\|_\infty\leq \epsilon$ for all $p\in \partial D$, and that $\|{\tilde \gamma^M}(O_D)-{\tilde \gamma^N}(O_D)\|_\infty=\epsilon$.  We'll show that then $\|{\tilde \gamma^M}(p)-{\tilde \gamma^N}(p)\|_\infty\leq \epsilon$ for all $p\in D$.  Applying this result once gives that $d_\infty(\gamma^M\circ \iota,\gamma^N\circ \iota)=\epsilon$, where $\iota:X^i\hookrightarrow X$ is the inclusion.  Applying the result a second time establishes the lemma.     

To show that $\|{\tilde \gamma^M}(p)-{\tilde \gamma^N}(p)\|_\infty\leq \epsilon$ for any $p\in D$, write $p=tO_D+(1-t)b$ for some $b\in \partial D$  and $0\leq t\leq 1$.   Since the restrictions of ${\tilde \gamma^M}$ and ${\tilde \gamma^N}$ to any radial line segment from $O_D$ to $\partial D$ are linear, we have that 
\begin{align*}
{\tilde \gamma^M}(p)&=t{\tilde \gamma^M}(O_D)+(1-t){\tilde \gamma^M}(b),\\
{\tilde \gamma^N}(p)&=t{\tilde \gamma^N}(O_D)+(1-t){\tilde\gamma^N}(b).  
\end{align*}
Thus
\begin{align*}
\|{\tilde \gamma^M}(p)-{\tilde \gamma^N}(p)\|_\infty &\leq t\|{\tilde \gamma^M}(O_D)-{\tilde \gamma^N}(O_D)\|_\infty+(1-t)\|{\tilde \gamma^M}(b)-{\tilde \gamma^N}(b)\|_\infty\\
&\leq t\epsilon+(1-t)\epsilon=\epsilon
\end{align*}
as needed.\end{proof}

\subsubsection{Part 3: Verifying that $M\cong H_i(\gamma^M)$ and $N\cong H_i(\gamma^N)$}\label{LastPartOfProof}
We'll now show that $M\cong H_i(\gamma^M)$; the argument that $N\cong H_i(\gamma^N)$ is essentially the same.  

\begin{lem}\label{Lem:keyFunctionProperty}
For any disk $D$ in $\tilde X-X'$ and $p\in D$, \[\tilde \gamma^M(p)\leq \tilde \gamma^M(O_D).\]  
\end{lem}

\begin{proof}
Since we have defined $\tilde \gamma^M$ to be linear along the radial line segments of disks in $\tilde X-X'$, it suffices to prove the result for $p\in \partial D$.

For $D=D_w$ an $i$-dimensional disk, \[\tilde \gamma^M(p)=\vv{\lfloor \gr(w) \rfloor} \leq \gr(w)\leq\tilde \gamma^M(O_{D_w}),\] so the result holds.  

For $D=D_y$ an $(i+1)$-dimensional disk, let $\longsup\in \R^n$ be given by
\[\longsup=\sup\left(\{\gr(w)\mid w\in \W_y\cap \W_1\}\cup \{\gr(w)+\vec\epsilon\mid w\in \W_y\cap \W_2\}\right).\]
It follows from properties 1 and 2 in our definition of $\sigma_y$ that $\tilde\gamma^M(p)\leq \longsup$.

To finish the proof of the lemma, we show that $\longsup\leq \tilde \gamma^M(O_{D_y})$.  If $y\in \Y_1$, then since $y\in \free[\W_1,\W_2(-\epsilon)]$, we have that $\gr(w)\leq \gr(y)$ for all $w\in \W_y\cap \W_1$, and $\gr(w)+\vec \epsilon\leq \gr(y)$ for all $w\in \W_y\cap \W_2$.  Thus \[\longsup\leq \gr(y)=\tilde \gamma^M(O_{D_y}).\]  

If $y\in \Y_2$, then since $y\in \free[\W_1(-\epsilon),\W_2]$, we have that \[\gr(w)\leq\gr(w)+\vec \epsilon\leq \gr(y)\leq\gr(y)+\vec\epsilon\] for all $w\in \W_1\cap \W_y$, and $\gr(w)+\vec\epsilon\leq \gr(y)+\vec\epsilon$ for all $w\in \W_2\cap \W_y$.  Thus \[\longsup\leq \gr(y)+\vec\epsilon=\tilde \gamma^M(O_{D_y}).\]  Since either $y\in \Y_1$ or $y\in \Y_2$, we therefore have that $\longsup\leq\gamma^M(O_{D_y})$, as desired.
\end{proof}

For $a\in \R^n$, let $\F_a$ denote the subcomplex of $X$ containing only those cells $e$ of $X$ entirely contained in $\FS(\gamma^M)_a$.  

\begin{lem}\label{Lem:Def_Retract}
${\F}_a$ is a deformation retract of $\FS(\gamma^M)_a$.  
\end{lem}

\begin{proof}
There exists a pair of deformation retractions \[\FS(\gamma^M)_a\to \F_a\cup (\FS(\gamma^M)_a\cap X^i)\to \F_a;\] this follows easily from \cref{Lem:keyFunctionProperty} and the fact that $\tilde \gamma^M$ is linear along the radial line segments of disks in $\tilde X-X'$.  Clearly, the composition of the two maps is a deformation retraction $\FS(\gamma^M)_a\to \F_a$.
\end{proof}

$\{\F_a\}_{a\in \R^n}$ defines an $n$-dimensional cellular filtration ${\F}$, and the inclusions ${\F}_a \hookrightarrow \FS(\gamma^M)_a$ define a morphism $\chi:\F\to \FS(\gamma^M)$ of $n$-dimensional filtrations.  By \cref{Lem:Def_Retract}, $H_i(\chi)_a:H_i {\F}_a\to H_i (\gamma^M)_a$ is an isomorphism for all $a\in \R$, so $H_i(\chi): H_i{\F}\to H_i(\gamma^M)$ is an isomorphism of $n$-modules.  Thus, to prove that $M\cong H_i(\gamma^M)$, it's enough to check that $M\cong H_i{\F}$.  This is a straightforward application of cellular homology, as we now explain.

Let $M'$ be the $n$-module with $M'_a=\free[\W_1,\W_2(-\epsilon)]_a/\langle \Y_1,\Y_2(\epsilon)\rangle_a$ and with $\varphi_{M'}(a,a')$ the map induced by the inclusion $\free[\W_1,\W_2(-\epsilon)]_a\hookrightarrow \free[\W_1,\W_2(-\epsilon)]_{a'}$.  $M'$ is canonically isomorphic to  $\free[\W_1,\W_2(-\epsilon)]/\langle \Y_1,\Y_2(-\epsilon)\rangle$, so since $M\cong\langle \W_1,\W_2(-\epsilon)|\Y_1,\Y_2(-\epsilon)\rangle$, $M$ is isomorphic to $M'$.  

To show that $M\cong H_i{\F}$, then, it's enough to define isomorphisms $\alpha_a:H_i{\F_a}\to M'_a$ for all $a\in \R^n$, such that the following diagram commutes whenever $a\leq b$: 
\begin{equation}\label{eq:alpha_naturality}
\begin{aligned}
\xymatrixcolsep{3.0pc}
\xymatrix{
H_i{\F}_a\ar[r]^{\varphi_{H_i{\F}}(a,b)}\ar[d]_{\alpha_a}^{\cong} &H_i{\F}_b\ar[d]^{\alpha_b}_{\cong}\\
M'_a\ar[r]_{\varphi_M(a,b)} &M'_b
}
\end{aligned}
\end{equation}

Define $\W_a\subset \W$ and $\Y_a\subset \Y$ by 
\begin{align*}
\W_a&=\{w\in \W_1\mid \gr(w)\leq a\} \cup \{w\in \W_2\mid \gr(w)+\vec\epsilon\leq a\}.\\
\Y_a&=\{y\in \Y_1\mid \gr(y)\leq a\} \cup \{y\in \Y_2\mid \gr(y)+\vec\epsilon\leq a\}.
\end{align*}

Let $\E=\{e^i_w\mid w\in \W_a\}.$  It follows from \cref{Lem:keyFunctionProperty} that $\E$ is exactly the set of $i$-cells in $\F_a$ which do not lie in $X'$.  The $i$-cells of $\F_a$ form a basis for $C_i(\F_a)$, so $\E$ is a linearly independent set in $C_i(\F_a)$.  In fact, it's easy to check that $\E$ is a basis for $\ker\partial^a_i$, where we have written $\partial^{\F_a}_i$ simply as $\partial^a_i$.  

Let $\V= \{x^{a-\gr(w)}w\mid w\in \W_a\}.$  $\V$ is a basis for $\free[\W_1,\W_2(-\epsilon)]_a$.  We have a bijection $\E\to \V$ sending $e^i_w$ to $x^{a-\gr(w)}w$.  Since $\E$ and $\V$ are bases for $\ker\partial^a_i$ and $\free[\W_1,\W_2(-\epsilon)]_a$, respectively, this bijection extends linearly to an isomorphism 

\[\tilde\alpha_a:\ker\partial^a_i\to \free[\W_1,\W_2(-\epsilon)]_a.\]  

We next show that $\tilde\alpha_a(\im \partial^a_{i+1})=\langle \Y_1,\Y_2(-\epsilon)\rangle_a$.  By \cref{Lem:keyFunctionProperty}, $\{e^{i+1}_y\mid y\in \Y_a\}$ is the set of $(i+1)$-cells in $\F_a$, hence is a basis for $C_{i+1}(\F_a)$.  Thus, by \cref{Lem:DegreeLemma},
\[\im \partial^a_{i+1}=\newspan\bigg\{\sum_{w\in \W_y} c(y,w) e^i_w \mid y\in \Y_a\bigg\}.\]
On the other hand, we have that \[\langle \Y_1,\Y_2(-\epsilon) \rangle_a=\newspan \bigg\{\sum_{w\in \W_y} c(y,w)x^{a-\gr(w)}w \mid y\in \Y_a\bigg\}.\]
It follows that $\tilde\alpha_a(\im \partial^a_i)=\langle \Y_1,\Y_2(-\epsilon)\rangle_a$, as desired.  

Hence, since the singular and cellular homology of CW-complexes are naturally isomorphic, $\tilde\alpha_a$ descends to an isomorphism $\alpha_a: H_i{\F}_a\to M'_a$.  It remains to check that for this definition of the maps $\alpha_a$, the diagram~(\ref{eq:alpha_naturality}) above commutes. 
For $a\leq b\in \R^n$, we have a commutative diagram of the following form, where the horizontal arrows denote the inclusions:
\[
\xymatrixcolsep{2.5pc}
\xymatrix{
\ker\partial^a_i\ar[d]_{\tilde \alpha_a}\ar@{^{(}->}[r] & \ker\partial^b_i\ar[d]^{\tilde\alpha_b}\\
\free[\W_1,\W_2(-\epsilon)]_a\ar@{^{(}->}[r]  & \free[\W_1,\W_2(-\epsilon)]_b
}
\]
It now follows by taking quotients that the diagram~(\ref{eq:alpha_naturality}) commutes.  This completes the proof of Proposition~\ref{RealizationProp}.

\subsection{Universality of the Bottleneck Distance}\label{Sec:UniversalityOfBottleneck}
We now present a universality result for $d_B$ analogous to our universality result Theorem~\ref{MainOptimality} for $d_I$.  It is most convenient to formulate this result in terms of reduced homology.  We let $\tilde H_i:\nfilt\to \nmod$ denote the $i^{th}$ reduced persistent homology functor, defined in the obvious way.

Let $\PFmod$ denote the full subcategory of $\onemod$ whose objects are \pfd persistence modules.  For $i\geq 0$, we say a pseudometric on $[\obj \PFmod]$ is {\it $i$-stable} if for any topological spaces $X,Y$ and functions 
\[\gamma^X:X\to \R,\quad \gamma^Y:Y\to \R\] 
such that $\tilde H_i(\gamma^X)$ and $\tilde H_i(\gamma^Y)$ are \pfd, we have
\[d(\tilde H_i(\gamma^X),\tilde H_i (\gamma^Y))\leq d_\infty(\gamma^X,\gamma^Y).\]

For $i\geq 0$, we say a pseudometric $d$ on $[\obj \PFmod]$ is {\it $i$-universal} if $d$ is $i$-stable and for any other pseudometric $d'$ on $[\obj \PFmod]$, $d'\leq d$.

\begin{thm}\label{PFOptimality}
For any field $k$ and $i\geq 0$, $d_B$ is $i$-universal.
\end{thm}

\begin{proof}
Using Proposition~\ref{PFRealizationProp} below in place of Proposition~\ref{RealizationProp}, the proof of Theorem~\ref{MainOptimality} carries over to give the result.
\end{proof}

Note that whereas Theorem~\ref{MainOptimality} holds for prime fields $k$ and $i\geq 1$, Theorem~\ref{PFOptimality} holds for arbitrary fields $k$ and $i\geq 0$.

\begin{prop}[Existence of Geometric Lifts of Interleavings for \pfd $1$-Modules]\label{PFRealizationProp} Let $k$ be any field, and let $M$ and $N$ be \pfd $1$-modules with $d_B(M,N)=\epsilon$.  Then for any $i\geq 0$ and $\delta>0$, there exists a CW-complex $X$ and continuous functions $\gamma^M,\gamma^N:X\to \R$ such that 
\[M\cong \tilde H_i(\gamma^M), \quad N\cong \tilde H_i(\gamma^N),\quad d_\infty(\gamma^M,\gamma^N)\leq\epsilon+\delta.\]\end{prop}

\begin{proof}  
An easy constructive proof, similar on a high level to our proof of Proposition~\ref{RealizationProp}, follows from the definition of $d_B$ and the structure theorem for 1-D persistence modules,~\cref{Thm:WCBStructureTheorem}.  We leave the details to the reader.   
\end{proof}

\begin{remark}\label{FrosiniRemark} In the special case that $i=0$, our Theorem~\ref{PFOptimality} generalizes the universality result of \cite{d2010natural}.
\end{remark}

  \section{The Closure Theorem}\label{Sec:ClosureTheorem}
This section is devoted to the proof our fourth and last main result:

\begin{thm}[The Closure Theorem]\label{InterleavingThm}If $M$ and $N$ are finitely presented $n$-modules and $d_I(M,N)=\epsilon$, then $M$ and $N$ are $\epsilon$-interleaved. \end{thm}    

\begin{cor}\label{MetricCorollary} $d_I$ restricts to a metric on isomorphism classes of finitely presented $n$-modules. \end{cor} 

\begin{remark} In the special case $n=1$, the closure theorem follows easily from the Isometry Theorem \ref{Thm:Isometry} and the observation that a finitely presented 1-D persistence module has a finite barcode with all intervals of the form $[s,t)$, $s<t\in \R\cup \{\infty\}$.  
\end{remark}

We prepare for the proof of the closure theorem with a few definitions and lemmas.

For $M$ any finitely presented $n$-module, let $U_M\subset \R^n$ be the set of grades of the generators and relations in some fixed, arbitrarily chosen, finite presentation for $M$; let $U_M^i\subset \R$ be the set of $i^{th}$ coordinates of the elements of $U_M$; and let $\bar U_M^i=U_M^i\cup\{ -\infty \}.$  

The proof of the following result is straightforward:
\begin{lem}\label{FirstIsomorphismLemma} For any $a\leq b\in \R^n$ such that $(a_i,b_i]\cap U_M^i=\emptyset$ for all $i$, $\varphi_M(a,b)$ is an isomorphism. \end{lem}


Let \[\fl_M:\R^n\to \Pi_{i=1}^n {\bar U_M^i}\] be defined by $\fl_M(a_1,\ldots,a_n)=(a'_1,\ldots,a'_n)$, where $a'_i$ is the largest element of $U_M^i$ such that $a'_i\leq a_i$, if such an element exists, and $a'_i=-\infty$ otherwise.

\begin{lem}\label{SecondConsequence}For any $a\in \R^n$ with $\fl_M(a)\in \R^n$, $\varphi_M(\fl_M(a),a)$ is an isomorphism. \end{lem}

\begin{proof} This is an immediate consequence of Lemma~\ref{FirstIsomorphismLemma}. \end{proof}   

\begin{lem}\label{FirstConsequence}For any finitely presented $n$-module $M$ and $a\in \R^n$, there exists $t\in (0,\infty)$ such that $\varphi_M(a,a+\vec s)$ is an isomorphism for all $0\leq s\leq t$. \end{lem}  

\begin{proof} This too is an immediate consequence of Lemma~\ref{FirstIsomorphismLemma}. \end{proof}

For the remainder of this section, we will write $a+\vec t$ simply as $a+t$ for any $a\in \R^n$ and $t\in \R$. 

\begin{proof}[Proof of \cref{InterleavingThm}]
Let $M$ and $N$ be finitely presented $n$-modules with $d_I(M,N)=\epsilon$.  By Lemma~\ref{FirstConsequence} and the finiteness of $U_M$ and $U_N$, there exists $\delta>0$ such that for all $\z\in U_M$, $\varphi_N(\z+\epsilon ,\z+\epsilon+\delta)$ and $\varphi_M(\z+2\epsilon ,\z+2\epsilon+2\delta)$ are isomorphisms, and for all $\z\in U_N$, $\varphi_M(\z+\epsilon,\z+\epsilon+\delta)$ and $\varphi_N(\z+2\epsilon,\z+2\epsilon+2\delta)$ are isomorphisms.     

By Remark~\ref{EasyInterleavingRemark}, since $d_I(M,N)=\epsilon$, $M$ and $N$ are $(\epsilon+\delta)$-interleaved.  

Theorem~\ref{InterleavingThm} then follows from the following lemma.
\end{proof}

\begin{lem}\label{SmallerInterleavingLemma}
Let $M$ and $N$ be finitely presented $n$-modules and suppose there exist $\epsilon\geq 0$ and $\delta>0$ such that
\begin{enumerate*}
\item $M$ and $N$ are $(\epsilon+\delta)$-interleaved,
\item for all $\z\in U_M$, $\varphi_N(\z+\epsilon,\z+\epsilon+\delta)$ and $\varphi_M(\z+2\epsilon,\z+2\epsilon+2\delta)$ are isomorphisms,
\item for all $\z\in U_N$, $\varphi_M(\z+\epsilon,\z+\epsilon+\delta)$ and $\varphi_N(\z+2\epsilon,\z+2\epsilon+2\delta)$ are isomorphisms.
\end{enumerate*}
Then $M$ and $N$ are $\epsilon$-interleaved.
\end{lem}
\begin{proof} Let $f:M\to N(\epsilon+\delta)$ and $g:N\to M(\epsilon+\delta)$ be interleaving morphisms. 

We define $\epsilon$-interleaving morphisms 
\begin{align*}
{\tilde f}&:M\to N(\epsilon)\\
{\tilde g}&:N\to M(\epsilon)
\end{align*}
by specifying $\tilde f_\z:M_\z\to N_{\z+\epsilon}$ and $\tilde g_\z:N_\z\to M_{\z+\epsilon}$ for each $\z\in \R^n$.

First, for $\z\in U_M$ define \[{\tilde f}_\z=\varphi^{-1}_N(\z+\epsilon,\z+\epsilon+\delta) \circ f_\z.\]  Then for arbitrary $\z\in \R^n$ such that $\fl_M(\z)\in \R^n$ define \[{\tilde f}_\z=\varphi_N(\fl_M(\z)+\epsilon,\z+\epsilon)\circ {\tilde f}_{\fl_M(\z)} \circ \varphi^{-1}_M(\fl_M(\z),\z).\]  (Note that $\varphi^{-1}_M(\fl_M(\z),\z)$ is well defined by Lemma~\ref{SecondConsequence}.)  Finally, for $\z\in \R^n$ such that $\fl_M(\z)\not\in \R^n$, define ${\tilde f}_\z=0$.  (If $\fl_M(\z)\not\in \R^n$ then $M_\z=0$, so this last part of the definition is reasonable.)

Symmetrically, for $\z\in U_N$ define \[{\tilde g}_\z=\varphi^{-1}_M(\z+\epsilon,\z+\epsilon+\delta) \circ g_\z,\] and for arbitrary $\z\in \R^n$ such that $\fl_N(\z)\in\R^n$ define \[{\tilde g}_\z=\varphi_M(\fl_N(\z)+\epsilon,\z+\epsilon)\circ {\tilde g}_{\fl_N(\z)} \circ \varphi^{-1}_N(\fl_N(\z),\z).\]  For $\z\in \R^n$ s.t. $\fl_N(\z)\not\in \R^n$, define ${\tilde g}_\z=0$.   

We need to check that ${\tilde f},{\tilde g}$ as thus defined are in fact morphisms.  We perform the check for ${\tilde f}$; the check for ${\tilde g}$ is the same.

For $\y\leq \b\in \R^n$ such that $\fl_M(\y)\not\in \R^n$, then since $M_\y=0$, it's clear that \[{\tilde f}_\b \circ \varphi_M(\y,\b)=0=\varphi_N(\y+\epsilon,\b+\epsilon)\circ {\tilde f}_\y.\]

For $\y\leq \b\in \R^n$ such that $\fl_M(\y)\in \R^n$, the equality \[{\tilde f}_\b \circ \varphi_M(\y,\b)=\varphi_N(\y+\epsilon,\b+\epsilon)\circ {\tilde f}_\y\] is immediate from the commutativity of the following diagram; in this diagram and those that follow, unlabeled edges represent transition maps, and edges labeled `$\cong$' represent the inverses of transition maps which are invertible by assumption.
\[
\xymatrix{
M_\y\ar[r]\ar@/_3.5pc/[dddd]_{\tilde f_\y}\ar[d]_\cong                                     &M_\b\ar@/^3.5pc/[dddd]^{\tilde f_\b}\ar[d]^\cong\\ 
M_{\fl_M(\y)}\ar[r]\ar[d]_{f_{\fl_M(\y)}}                                  &M_{\fl_M(\b)}\ar[d]^{f_{\fl_M(\b)}}\\  
N_{\fl_M(\y)+\epsilon+\delta}\ar[r]\ar[d]_\cong           &N_{\fl_M(\b)+\epsilon+\delta}\ar[d]^\cong\\  
N_{\fl_M(\y)+\epsilon}\ar[r]\ar[d]                           &N_{\fl_M(\b)+\epsilon}\ar[d]\\  
N_{\y+\epsilon}\ar[r]                              &N_{\b+\epsilon}
}
\]

To finish the proof of the lemma, we need to check that ${\tilde g}(\epsilon)\circ {\tilde f}=\varphi_M^{2 \epsilon}$ and ${\tilde f}(\epsilon)\circ {\tilde g}=\varphi_N^{2 \epsilon}$.  We perform the first check;  the second check is the same.

For $\z\in \R^n$, if $\fl_M(\z)\not\in \R^n$ then since $M_\z=0$, \[{\tilde g}_{\z+\epsilon}\circ {\tilde f}_\z=0=\varphi_M(\z,\z+2\epsilon).\]

The verification that this also holds for $\z$ with $\fl_M(\z)\in \R^n$ is a large diagram chase, which we break up into two smaller diagram chases: We first verify the result for $\z\in U_M$.  We'll then use this special case to verify the result for arbitrary $\z\in \R^n$ with $\fl_M(\z)\in \R^n$.  

For $\z\in U_M$ we obtain the result from the commutativity of the following diagram:
\[
\xymatrix{
M_\z\ar[rr]^{f_\z}\ar[ddr]\ar@/^2pc/[rrr]^{\tilde f_\z} \ar@/_5.5pc/[ddddrrr] & &N_{\z+\epsilon+\delta}\ar[ddl]^{g_{\z+\epsilon+\delta}}\ar[r]^{\cong} &N_{\z+\epsilon} \ar[ddl]_{g_{\z+\epsilon}}\ar@/^5pc/[dddd]^{\tilde g_{\z+\epsilon}}\ar[d]^{\cong}\\
& & & N_{\fl_N(\z+\epsilon)}\ar[d]^{g_{\fl_N(\z+\epsilon)}}\\ 
& M_{\z+2\epsilon+2\delta}\ar[ddrr]_\cong &M_{\z+2\epsilon+\delta}\ar[l] &M_{\fl_N(\z+\epsilon)+\epsilon+\delta}\ar[l]\ar[d]^{\cong}\\  
& & & M_{\fl_N(\z+\epsilon)+\epsilon}\ar[d]\\
& & & M_{\z+2\epsilon}
}
\]

Then, for arbitrary $\z\in \R^n$ with $\fl_M(\z)\in \R^n$, we have, using that ${\tilde g}$ is a morphism, that the following diagram commutes:
\[
\xymatrixcolsep{3pc}
\xymatrix{
M_\z\ar@/^2.5pc/[rrr]^{\tilde f_\z}\ar@/_1.5pc/[rrrdd]\ar[r]^\cong &M_{\fl_M(\z)}\ar[r]^{\tilde f_{\fl_M(\z)}}\ar[dr] &N_{\fl_M(\z)+\epsilon}\ar[r]\ar[d]^{\tilde g_{\fl_M(\z)+\epsilon}} &N_{\z+\epsilon}\ar[dd]^{\tilde g_{\z+\epsilon}} \\
& & M_{\fl_M(\z)+2\epsilon}\ar[dr]\\
& & & M_{\z+2\epsilon}
}
\]
This gives that ${\tilde g}_{\z+\epsilon}\circ {\tilde f}_\z=\varphi_M(\z,\z+2\epsilon )$, as we wanted.\end{proof}


  \section{Discussion}\label{DiscussionSection}

\nparagraph{Computation} 
The results of this paper establish that $d_I$ is, in several senses, a very well behaved generalization of $d_B$ to the multidimensional setting.  
Insofar as $d_I$ is in fact a good choice of distance on multidimensional persistence modules, the question of if and how it can be computed or approximated is interesting and potentially significant from the standpoint of applications.  As noted in the introduction, this question remains open.    One potential application of interleaving distance computations is to shape matching, where distances between multidimensional persistence modules have already been applied~\cite{biasotti2008multidimensional}.  I also imagine that in statistical settings, computation of $d_I$ could be useful in resampling methods for computing confidence regions for estimates of multi-D persistent homology.

\nparagraph{Questions Related to Universality}

Our investigation of the universality properties of $d_I$ raises several questions:

\begin{enumerate}
\item Our universality result Theorem~\ref{MainOptimality} demonstrates that $d_I$ is $i$-universal, as defined in Section~\ref{Sec:OptimalityDefs}, when $k$ is a prime field and $i\geq 1$.  We have hypothesized (Conjecture~\ref{Conj:ZeroOptimality}) that in fact $d_I$ is $i$-universal for arbitrary $k$ and $i \geq 0$.  Can we prove this?
\item Our definition of universality is induced by a particular choice of definition of the stability of a pseudometric on $[\obj \nmod]$, given in Section~\ref{MultidimensionalStabilitySection}.  By varying our definition of stability, we obtain different definitions of universality, and thus are led to a number of interesting questions about the universality of pseudometrics analogous to those considered in this paper.  To give one example, say a pseudometric $d$ on $[\obj \onemod]$ is {\it GH-stable} if for each $i\geq 0$ and pair of finite metric spaces $X,Y$, \[d(H_i \Rips(X), H_i \Rips(Y))\leq d_{GH}(X,Y),\] where $d_{GH}$ denotes the Gromov-Hausdorff distance.  Say a pseudometric $d$ is {\it GH-universal} if it is GH-stable and for every GH-stable pseudometric $d'$, $d'(M,N)\leq d(M,N)$ for all $1$-modules $M,N$ such that $\exists$ $i\geq 0 $ with $M,N\in \im H_i\Rips(\cdot)$.
It was shown in \cite{chazal2009gromov} that $d_I$ is GH-stable.  Is it true that $d_I$ is GH-universal?

\item\label{Generalizations} Can we obtain results analogous to our universality result Theorem~\ref{MainOptimality} for more general types of persistent homology modules?  For instance, can we prove a result analogous to Theorem~\ref{MainOptimality} for levelset zig-zag persistence \cite{carlsson2009zigzag}?

\item A question related to question~\ref{Generalizations}: is there a way of algebraically reformulating the bottleneck distance for zig-zag persistence modules as an analogue of $d_I$ in such way that the definition generalizes to a larger classes of commutative quiver representations \cite{derksen2005quiver}?
\end{enumerate}

\subsection*{Acknowledgments}
The first version of this paper 
was written while I was a graduate student.  Discussions with my Ph.D. adviser Gunnar Carlsson catalyzed the research presented here in several ways.  In addition, Gunnar served as a patient and helpful sounding board for the ideas of this paper.  I thank him for his support and guidance.

Thanks to Henry Adams, Peter Bubenik, Patrizio Frosini, Peter Landweber, Dmitriy Morozov, and the anonymous referees for useful corrections and helpful feedback on this work.  

Parts of the exposition in \cref{Sec:Background,Sec:IsometrySection} benefited from edits done jointly with Ulrich Bauer on closely related material in \cite{bauer2014induced}.  

The main result of William Crawley-Boevey's paper \cite{crawley2012decomposition} plays an important role in the present version of this work.  
I thank Bill for writing his paper and both Bill and Vin de Silva for enlightening discussions about structure theorems for $\R$-graded persistence modules.  

Thanks to Stanford University, the Technion, the Institute for Advanced Study, and the Institute for Mathematics and its Applications for their support  hospitality during the writing and revision of this paper.  This work was supported by ONR grant N00014-09-1-0783 and NSF grant DMS-1128155.  Any opinions, findings, and conclusions
or recommendations expressed in this material are those of the author and do not necessarily reflect the views of the National Science Foundation.

\bibliographystyle{plain}	
\bibliography{Revised_Paper_Refs}	

 \end{document}